\documentclass[twoside,11pt,final,dvipsnames]{entics} 

\def\lastname{Sokolova, Woracek}
\def\runtitle{Convex algebras with monotone semicontinuous operations}
\def\copynames{A. Sokolova, H.Woracek}

\usepackage{enticsmacro}
\usepackage{amsmath}
\usepackage{amssymb}
\usepackage{arydshln}
\usepackage{bm}
\usepackage{dsfont}
\usepackage{enumitem}
\usepackage[OT2,OT1]{fontenc}
\usepackage[cal=cm,scr=boondoxo]{mathalfa}
\usepackage{pifont}
\usepackage{scalerel}
\usepackage{stmaryrd}
\usepackage{textcomp}
\usepackage{tikz}
	\usetikzlibrary{arrows}
	\usetikzlibrary{patterns}
	\usetikzlibrary{angles}
\usepackage{tikz-cd}
\usepackage[textwidth=3cm,colorinlistoftodos]{todonotes}

\newcounter{Enum}				
\newenvironment{Enumerate}{\begin{enumerate}[label={\rm({\roman*})}]}{\end{enumerate}}

\newcommand{\descriptionlabelsave}{}		
\newenvironment{Itemize}{%
	\renewcommand{\descriptionlabelsave}{\descriptionlabel}\renewcommand{\descriptionlabel}{$\triangleright$}%
	\begin{description}[leftmargin=15pt,itemindent=-5.2pt]}{%
	\end{description}\renewcommand{\descriptionlabel}{\descriptionlabelsave}}

\newcounter{StepsCount}				
\newenvironment{Steps}{%
	\begin{list}{{\sf Step}\ \ding{\value{StepsCount}}\,:}{%
	\usecounter{StepsCount} \leftmargin=0pt \labelwidth=46pt \itemindent=\labelwidth%
	\itemsep=5pt\listparindent=\parindent} \setcounter{StepsCount}{191}}{\end{list}}

\newcounter{ElistCount}				
\newenvironment{Elist}{%
	\begin{list}{{\rm(\roman{ElistCount})}}{%
	\usecounter{ElistCount} \leftmargin=0pt \labelwidth=0pt \itemindent=6pt%
	\itemsep=5pt\listparindent=\parindent} \setcounter{ElistCount}{0}}{\end{list}}

\newenvironment{Ilist}{
	\begin{list}{$\triangleright$}{\leftmargin=0pt \labelwidth=11pt \itemindent=\labelwidth%
	\itemsep=5pt\listparindent=\parindent}}{\end{list}}

\newcommand{\mc}[1]{{\mathcal{#1}}}			
\newcommand{\bb}[1]{{\mathbb{#1}}}			

\newcommand{\DS}{\mid\mkern3mu}				
\newcommand{\DSb}{\mkern4.5mu\big|\mkern7.5mu}		
\newcommand{\DSB}{\mkern4.5mu\Big|\mkern7.5mu}		
\newcommand{\DQ}{\mkern6mu}				
\newcommand{\DP}{{\mathop:\kern5pt}}			
\newcommand{\DF}{\colon}				
\newcommand{\DE}{\mathrel{\mathop:}=}			
\newcommand{\DI}{\mathrel{\mathop:}\Leftrightarrow}	

\newcommand{\CAS}{&\text{if}\ }				
\newcommand{\CASO}{&\text{otherwise}}			

\DeclareMathOperator{\Id}{id}				

\newcommand{\LC}{{\sf($\overline{\text{LC}}$)}}
\DeclareMathOperator*{\Bigplus}{\scalerel*{+}{\sum}}
\DeclareMathOperator{\Aut}{Aut}				
\newcommand{\Eats}{
	\raisebox{-1pt}{$\mkern3mu
	\begin{tikzpicture}[scale=0.12]
		\draw[thin] (0,0) -- (30:1cm) arc (30:330:1cm) -- cycle;
		\fill (0,0.55) circle (1.5mm);
	\end{tikzpicture}
	\mkern4mu$}
}

\def\conf{MFPS 2026}
\volume{6}
\def\volu{6}
\def\papno{17}
\def\mydoi{19317}

\begin{document}

\begin{frontmatter}
	\title{Convex Algebras on an Interval with Semicontinuous Monotone Operations}
	\author{Ana Sokolova\thanksref{a}\thanksref{myemail}}
	\author{Harald Woracek\thanksref{b}\thanksref{coemail}}
	\address[a]{Department of Computer Science\\ University of Salzburg\\ Salzburg, Austria}
	\thanks[myemail]{Email: \href{mailto:ana.sokolova@cs.uni-salzburg.at} 
		{\texttt{\normalshape ana.sokolova@cs.uni-salzburg.at}}} 
	\address[b]{Institute for Analysis and Scientific Computing\\ TU Wien\\ Vienna, Austria} 
	\thanks[coemail]{Email: \href{mailto:harald.woracek@tuwien.ac.at}
		{\texttt{\normalshape harald.woracek@tuwien.ac.at}}} 
	\begin{abstract} 
		In a recent work of Matteo Mio on compact quantitative equational theories
		(here compact means that all its consequences are derivable by means of finite proofs) 
		convex algebras on the carrier set $[0,1]$ whose operations are monotone and satisfy certain 
		semicontinuity properties occurred. 
		We fully classify those algebraic structures by giving an explicit construction of all possible convex
		operations on $[0,1]$ possessing the mentioned properties. 
		Our result thus describes exactly the range of theories to which Mio's theorem applies.
	\end{abstract}
\end{frontmatter}

\section{Introduction} 

Convex algebras, also called barycentric algebras, semiconvex spaces or convex sets, are the algebras for probabilistic choice. 
They have several different presentations, the most common one being a carrier set together with infinitely many $(0,1)$-indexed binary
operations $\oplus_p$, which represent binary convex combinations. They have been studied for decades in convex geometry, universal
algebra, even economy, and lately extensively in different aspects of semantics for probabilistic systems and 
programs, c.f.~\cite{doberkat:2006,doberkat:2008,jacobs:2010,silva.sokolova:2011,mardare.panangaden.plotkin:2016,Bonchi0S17,BonchiSV19,CirsteaMNS0S25} to name a few,
including our work describing congruences of convex algebras~\cite{sokolova.woracek:pcacon} or 
termination~\cite{sokolova.woracek:CA-star_journal}.

A particularly interesting line of work has been developed in the last decade starting from the seminal paper on quantitative
equational theories~\cite{mardare.panangaden.plotkin:2016} 
by Mardare, Panangaden, and Plotkin, followed by the work on monads (for probability and nondeterminism) on metric 
spaces~\cite{mio.vignudelli:2020,mio.sarkis.vignudelli:2021,mio.sarkis.vignudelli:2022,mio.sarkis.vignudelli:2024} 
and their quantitative algebraic theories. One somewhat troublesome aspect of quantitative equational theories is an infinitary axiom,
which is (in some form) necessary for the soundness and completeness of quantitative
equational theories, yet problematic as it leads to countably branching proofs.

Recently, Mio introduced compact quantitative equational theories \cite{mio:2025} 
that enable proofs without the infinitary axiom. Notably, he proves that the quantitative theory of \emph{interpolative convex
algebras} (ICA) is compact. This theory is given by the axioms for convex algebras and an additional (quantitative) 
axiom known to axiomatize the Kantorovich lifting of a distance on a set $X$ to distributions on $X$.  
In his work, Mio generalizes ICA and the notion of Kantorovich distance to arbitrary convex algebras on $[0,1]$ and shows that
for certain convex algebras on $[0,1]$ --- these are convex algebras with
monotone and (in a certain sense) semicontinuous operations --- the generalized ICA theory is compact. 
Intrigued by Mio's wish to understand such algebras, their properties, and their shape, we started this work. 

In this paper, we give an explicit construction that describes the class of all convex algebras on $[0,1]$ with monotone and (in our
sense) semicontinuous operations. We point out that Mio had already identified crucial examples of such algebras which provide the
building blocks for any other algebra in this class. 

The convex algebras with monotone and semicontinuous operations have an interesting
structure: they are fully determined by a closed subset $E$ of $[0,1]$ which we informally call "eaters", as well as 
an "endpoint" attached to each connected component of 
(i.e., maximal open interval in) the complement $[0,1]\setminus E$ which is either $1$ or $\infty$. An
element $y$ eats an element $x$ if $y \oplus_p x = y$ for one (and hence every) $p\in(0,1)$. 
The eaters are the elements $y$ that eat the zero, and as a consequence of
the operation properties also eat anything in $[0,y)$. The convex combinations of elements in a connected component 
$(a,b)$ of $[0,1]\setminus E$ behave like linear combinations: If the attached endpoint is $1$, isomorphic to the standard $+_p$ on
$(0,1)$; if it is $\infty$, isomorphic to the standard $+_p$ on $(0,\infty)$. 
Note that $(0,1)$ and $(0,\infty)$ with the standard operations are not isomorphic.

The first main result is Theorem~\ref{E15} where we show how to construct a convex algebra on $[0,1]$ with monotone and semicontinuous
operations from given data which specifies the eaters and endpoints to be. Here, the crucial tool is the Plonka sum, going back
to~\cite{plonka:1967}, which is a general method from universal algebra.
Our second main result, is Theorem~\ref{E28} where we show that every convex algebra with monotone and semicontinuous operations is of
this form. This theorem is based on a detailed study of the structure of such algebras. 
Our third main result is Theorem~\ref{E40}. It characterises whether two algebras are isomorphic in terms of their respective sets of
eaters and endpoints.
Together, these three theorems yield a full classification.
Finally, in Proposition~\ref{E26}, we make the connection to Mio's conditions showing that we consider indeed the same class of
algebras.

The structure of the paper is as follows. Section~\ref{E54} is of preliminary nature:
there we introduce and discuss monotonicity and semicontinuity properties in convex algebras on $[0,1]$ 
(with the standard ordering and topology). 
In Section~\ref{E34} we present the important examples which provide the building blocks of our general construction.
Section~\ref{E55} is dedicated to our main results: Theorems~\ref{E15},~\ref{E28},~\ref{E40}. 
The proofs of Theorem~\ref{E15} and Theorem~\ref{E40} are the technically most involved part of the paper.
In Section~\ref{E22} we discuss the connection of our identified properties and the semicontinuity property from \cite{mio:2025}. 
\begin{ack}
	We thank Matteo Mio who brought up the question about the structure of convex algebras on $[0,1]$ with monotonicity and
	semicontinuity properties. Without him asking, this paper would probably not exist.
\end{ack}

\section{Monotonicity and continuity properties in convex algebras}
\label{E54}
\subsection{Convex algebras}

We start with recalling definitions and facts about convex algebras which are relevant for the present paper.

\begin{definition}
\label{E3}
	A \emph{convex algebra} is a set $X$ together with a family $\{\oplus_p\DS p\in(0,1)\}$ of binary operations on $X$ that 
	satisfy
	\begin{Ilist}
	\item[] ${\forall x\in X,p\in(0,1)\DP x\oplus_px=x}$
		\hspace*{62mm}(\emph{idempotence})
	\item[] ${\forall x,y\in X,p\in(0,1)\DP x\oplus_py=y\oplus_{1-p}x}$
		\hspace*{46.3mm}(\emph{parametric commutativity})
	\item[] ${\forall x,y,z\in X,p,q\in(0,1)\DP (x\oplus_py)\oplus_qz=x\oplus_{pq}\Big(y\oplus_{\frac{(1-p)q}{1-pq}}z\Big)}$
		\hspace*{10mm}(\emph{parametric associativity})
	\end{Ilist}
	We write a convex algebra as $\langle X,\oplus_p\rangle$.
\end{definition}

\noindent
Convex algebras have been invented as a generalisation of the following example from linear algebra.

\begin{example}
\label{E25}
	Let $V$ be a vector space over the scalar field $\bb R$, and let $X$ be a convex subset of $V$. We define operations 
	$+_p$ on $X$ as 
	\begin{equation}
	\label{E27}
		x+_py\DE px+(1-p)y\quad\text{for }x,y\in X,p\in(0,1),
	\end{equation}
	and refer to those as the \emph{operations induced by linear combinations}.
	Then $\langle X,+_p\rangle$ is a convex algebra. 

	Those convex algebras $\langle X,\oplus_p\rangle$ that are isomorphic to an algebra of this form are characterised 
	by the validity of the cancellation law
	\[
		\forall x,y,z\in X\DQ\forall p\in(0,1)\DP x\oplus_p z=y\oplus_p z\ \Rightarrow\ x=y
	\]
	This fact goes back to \cite{stone:1949,kneser:1952}.
\end{example}

\noindent
Convex algebras can be axiomatised in different ways, see e.g.\ \cite{flood:1981,romanowska.smith:1990,sokolova.woracek:pcacon}. 
One very practical fact is that the signature can be enlarged by using the following formula as a definition of two additional 
operations $\oplus_0$ and $\oplus_1$:
\begin{Ilist}
\item[] ${\forall x,y\in X\DP x\oplus_1y=x\wedge x\oplus_0y=y}$
	\hspace*{15mm}(\emph{projection axiom})
\end{Ilist}
Doing so, the parametric commutative and associative laws remain valid.

One can extend the signature even further: one may take the set 
\[
	\Big\{(p_i)_{i=1}^n\in\bb R^n\DSB n\geq 1,p_i\geq 0,\sum_{i=1}^np_i=1\Big\}
\]
as operations symbols, assign to each $(p_i)_{i=1}^n$ an $n$-ary operation and require a \emph{projection axiom} and a 
\emph{barycenter axiom}, cf.\ \cite[Definition~3.1]{sokolova.woracek:pcacon}.
The connection with Definition~\ref{E3} is that $\oplus_p$ corresponds to the binary operation with symbol $(p,1-p)$.
Usually the $n$-ary operation corresponding to $(p_i)_{i=1}^n$ is written as a formal convex combination:
\[
	(x_1,\ldots,x_n)\mapsto\bigoplus_{i=1}^n p_ix_i.
\]
In convex algebras as in Example~\ref{E25}, i.e.\ with convex operations $+_p$ induced by linear combinations,
this formal convex combination $\Bigplus_{i=1}^n p_ix_i$ equals the actual linear combination provided by the surrounding vector space, 
i.e., 
\[
	\Bigplus_{i=1}^n p_ix_i=\sum_{i=1}^n p_ix_i.
\]

\begin{remark}
\label{E16}
	Enlarging the signature proves useful when describing subalgebras generated by a given subset of a convex algebra.  In
	fact, if $\langle X,\oplus_p\rangle$ is a convex algebra and $M\subseteq X$, then the smallest subalgebra of $X$
	containing $M$ is 
	\[
		\Big\{\bigoplus_{i=1}^np_ix_i\DSB n\geq 1,p_i\geq 0,\sum_{i=1}^np_i=1,x_i\in M\Big\}.
	\]
	This also gives rise to a construction of free convex algebras. Let $A$ be a nonempty set. Consider the vector space
	$\bb R^A$, and denote by $e_a\in\bb R^A$ the vector whose $a$-th component is equal to $1$ while all others are $0$ (the
	\emph{canonical basis vectors}). Then the free convex algebra with basis $A$ is the subalgebra of $\bb R^A$ generated by
	the set $\{e_a\DS a\in A\}$. Explicitly, this is 
	\begin{align*}
		\mc DA= &\, 
		\Big\{\bigoplus_{i=1}^np_ie_{a_i}\DSB n\geq 1,p_i\geq 0,\sum_{i=1}^np_i=1,a_i\in A\Big\}
		\\
		= &\, 
		\Big\{(p_a)_{a\in A}\in\bb R^A\DSB p_a\geq 0\text{ with only finitely many $p_a$ nonzero},
		\sum_{a\in A}p_a=1\Big\}.
	\end{align*}
	In this context it is practical to use the following notation. Let $\langle X,\oplus_p\rangle$ be a convex algebra, let
	$(p_a)_{a\in A}\in\mc DA$, and let $(x_a)_{a\in A}\in X^A$. Write $\{a\in A\DS p_a>0\}=\{a_1,\ldots,a_n\}$ and set 
	\[
		\bigoplus_{a\in A}p_ax_a\DE\bigoplus_{i=1}^np_{a_i}x_i.
	\]
	Using this notation, the action of the 
	homomorphic extension $\varphi^\#\DF\mc DA\to X$ of a map $\varphi\DF A\to X$ mapping the set 
	$A$ into some convex algebra $\langle X,\oplus_p\rangle$ can be described conveniently as 
	\[
		\varphi^\#\big((p_a)_{a\in A}\big)=\bigoplus_{a\in A}p_a\varphi(a).
	\]
\end{remark}

\begin{remark}
\label{E30}
	The free convex algebra with two generators $\mc D2$ can be identified with $\langle [0,1],+_p\rangle$ where $+_p$ is 
	as in \eqref{E27}. If $\langle X,\oplus_p\rangle$ is any convex algebra and $\varphi\DF\{0,1\}\to X$ is a map,
	then $\varphi$ admits a unique extension to a homomorphism of $[0,1]$ into $X$, namely the map $\varphi^\#$ acting as 
	\begin{equation}
	\label{E36}
		\varphi^\#(t)\DE\varphi(1)\oplus_t\varphi(0)\qquad\text{for }t\in[0,1].
	\end{equation}
	The congruence lattice of $\langle [0,1],+_p\rangle$ is very simple. It has exactly five elements, namely 
	\begin{Enumerate}
	\item the diagonal $\{(t,t)\DS t\in[0,1]\}$, 
	\item $\{(0,0)\}\cup\Big[(0,1)\times(0,1)\Big]\cup\{(1,1)\}$,
	\item $\Big[[0,1)\times[0,1)\Big]\cup\{(1,1)\}$,
	\item $\{(0,0)\}\cup\Big[(0,1]\times(0,1]\Big]$,
	\item the universal relation $[0,1]\times[0,1]$.
	\end{Enumerate}
	This is proven, e.g., in \cite[Proposition~3.5]{pumpluen.roehrl:1990} or \cite[Example~4.13]{sokolova.woracek:pcacon}.
\end{remark}

\noindent
Homomorphisms of the form \eqref{E36} appear frequently, and we introduce a notation for them.

\begin{definition}
\label{E21}
	Let $\langle X,\oplus_p\rangle$ be a convex algebra, and let $x,y\in X$. Then we denote 
	\begin{equation}
	\label{E29}
		\Gamma_{x,y}\DF\left\{
		\begin{array}{rcl}
			[0,1] & \to & X
			\\
			t & \mapsto & y\oplus_tx
		\end{array}
		\right.
	\end{equation}
\end{definition}

\noindent
The intuition is that $\Gamma_{x,y}$ is a path starting at $x$ for $t=0$ and ending at $y$ for $t=1$ 
that parameterises the line segment connecting $x$ with $y$.

Since $\Gamma_{x,y}$ is a homomorphism, the kernel of $\Gamma_{x,y}$ as a relation on $X$ is one of the five congruences listed in 
Remark~\ref{E30}. 

\begin{definition}
\label{E5}
	Let $\langle X,\oplus_p\rangle$ be a convex algebra, and let $x,y\in X$. Then we define a relation $\Eats$ on $X$ as 
	\[
		y\Eats x\DI \forall t\in(0,1]\DP y\oplus_tx=y.
	\]
\end{definition}

\noindent
We see that $y\Eats x$ if and only if $\Gamma_{x,y}$ is constant on $(0,1]$, and in turn if and only if all of $(0,1]$ belongs to one
class of $\ker\Gamma_{x,y}$. 
Since none of the congruences in Remark~\ref{E30}(i)-(iii) identifies $1$ with any other point while those from 
Remark~\ref{E30}(iv)-(v) identify all of $(0,1]$, we have $y\Eats x$ if and only if $\ker\Gamma_{x,y}$ equals one of the 
congruences from Remark~\ref{E30}(iv)-(v). 
Moreover, 
\begin{equation}
\label{E37}
	y\Eats x\ \Leftrightarrow\ \exists t\in(0,1)\DP y\oplus_tx=y.
\end{equation}
We point out that the relation $\Eats$ contains a lot of information about the structure of the convex algebra, and plays a central
role in the present paper. 

\begin{remark}
\label{E20}
	The equivalence \eqref{E37} could also be deduced from an implication that holds in 
	any convex algebra $\langle X,\oplus_p\rangle$:
	\[
		\forall x,y,z\in X\DP
		\big(\exists p\in(0,1)\DP x\oplus_pz=y\oplus_pz\big)
		\Rightarrow\big(\forall p\in(0,1)\DP x\oplus_pz=y\oplus_pz\big).
	\]
	This implication can be shown with a direct argument that does not refer to the description of the congruence lattice 
	of $\langle [0,1],+_p\rangle$, see e.g.\ \cite[p.91]{koenig:1986} or \cite[Lemma~2.8]{keimel.plotkin:2017}.
\end{remark}

\noindent
Sometimes it is practical to use the following computation rule which 
is shown using the parametric commutative and associative laws; we skip the details. 

\begin{lemma}
\label{E18}
	Let $x_1,\ldots,x_n\in X$ and $p_1,\ldots,p_n\geq 0$ with $\sum_{i=1}^np_i=1$. Let $j\in\{1,\ldots,n\}$ be such that 
	$p_j>0$, and set 
	\[
		y_i\DE
		\begin{cases}
			x_j \CAS x_j\Eats x_i,
			\\
			x_i \CASO.
		\end{cases}
	\]
	Then 
	\[
		\bigoplus_{i=1}^np_ix_i=\bigoplus_{i=1}^np_iy_i.
	\]
\end{lemma}

\noindent
We will use the following facts about subalgebras of $\langle\bb R,+_p\rangle$, and make them explicit for later reference.
They follow since a convex map between intervals has an extension to an affine map on $\bb R$; we skip the details.

\begin{remark}
\label{E14}
	Let $t>0$. We denote by $M_t\DF\bb R\to\bb R$ the map acting as multiplication by $t$, i.e., 
	$M_t(x)\DE tx$. The relevance of these maps in the present context is that they describe homomorphisms between 
	intervals. 
	\begin{Ilist}
	\item If $\tau,\tau'\in(0,\infty)$, then there exists a unique isomorphism of $\langle[0,\tau),+_p\rangle$ onto 
		$\langle[0,\tau'),+_p\rangle$. Namely, the map $M_{\frac{\tau'}\tau}$. 
	\item The set of all automorphisms of $\langle[0,\infty),+_p\rangle$ is 
		$\{M_t\DS t>0\}$.
	\item If $\tau,\tau'\in(0,\infty)$, then there exist exactly two isomorphisms of $\langle[0,\tau],+_p\rangle$ onto 
		$\langle[0,\tau'],+_p\rangle$. Namely, the maps $M_{\frac{\tau'}\tau}$ and 
		$x\mapsto\tau'-M_{\frac{\tau'}\tau}(x)$. 
	\end{Ilist}
	Furthermore we have:
	\begin{Ilist}
	\item For any $\tau\in(0,\infty)$ the algebras $\langle[0,\tau),+_p\rangle$ and $\langle[0,\infty),+_p\rangle$ are 
		not isomorphic.
	\item No two algebras $\langle[0,\tau],+_p\rangle$ with $\tau\in(0,\infty)$ and $\langle[0,\tau'),+_p\rangle$ with 
		$\tau'\in(0,\infty]$ are isomorphic. 
	\end{Ilist}
\end{remark}

\subsection{Monotone operations}

We consider convex algebras that are additionally endowed with a compatible order. 

\begin{definition}
\label{E1}
	Let $\langle X,\oplus_p\rangle$ be a convex algebra and let $\preceq$ be a partial order on $X$. 
	We consider the following property:
	\begin{itemize}
	\item[\hspace*{2mm}{\sf(MO)}\hspace*{-2mm}] \hspace*{2mm}
		${\displaystyle
			\forall x,x',y\in X\DQ\forall p\in(0,1)\DP
			x\preceq x'\ \Rightarrow\ x\oplus_py\preceq x'\oplus_py
		}$
	\end{itemize}
	If {\sf(MO)} holds, we call the triple $\langle X,\oplus_p,\preceq\rangle$ a \emph{monotone convex algebra}. 
	If the order $\preceq$ is clear from the context, which will be the case for most of the paper, we drop
	$\preceq$ from the notation and speak of the monotone convex algebra $\langle X,\oplus_p\rangle$.
\end{definition}

\begin{example}
\label{E32}
	Let $V$ be an ordered vector space and $X$ a convex subset of $V$. Let $+_p$ be the operations on $X$
	induced by linear combinations and let $\preceq$ be the order on $X$ inherited from $V$. 
	Then $\langle X,+_p,\preceq\rangle$ is a monotone convex algebra.

	A simple instance of such a situation is $\mc D2$; just use the identification from Remark~\ref{E30}.
\end{example}

\noindent
In conjunction with the computation rules that are valid in every convex algebra, the monotonicity axiom {\sf(MO)} implies 
stronger properties. Notably, a monotonicity property of $\oplus_p$ in the parameter $p\in[0,1]$ follows automatically.

\begin{lemma}
\label{E2}
	Let $\langle X,\oplus_p,\preceq\rangle$ be a monotone convex algebra. Then we have
	\begin{Enumerate}
	\item 
		${\displaystyle
			\forall x,x',y,y'\in X\DQ\forall p\in[0,1]\DP
			x\preceq x'\wedge y\preceq y'\ \Rightarrow\ x\oplus_py\preceq x'\oplus_py'
		}$,
	\item
		${\displaystyle
			\forall x,y\in X\DQ\forall p\in[0,1]\DP
			x\preceq y\ \Rightarrow\ x\preceq x\oplus_py\preceq y
			\,\wedge\, x\preceq y\oplus_px\preceq y
		}$,
	\item
		${\displaystyle
			\forall x,y\in X\DQ\forall p,p'\in[0,1]\DP
			x\preceq y\wedge p\leq p'\ \Rightarrow\ y\oplus_px\preceq y\oplus_{p'}x
		}$.
	\end{Enumerate}
\end{lemma}
\begin{proof}
	\phantom{}
	\begin{Elist}
	\item Assume we have $x,x',y,y'\in X$ with $x\preceq x'$ and $y\preceq y'$, and $p\in[0,1]$. If $p=0$, then 
		$x\oplus_py=y\preceq y'=x'\oplus_py'$. If $p=1$, then $x\oplus_py=x\preceq x'=x'\oplus_py'$.
		For $p\in(0,1)$, we have 
		\[
			x\oplus_py\preceq x'\oplus_py=y\oplus_{1-p}x'\preceq y'\oplus_{1-p}x'=x'\oplus_py'.
		\]
	\item If $x,y\in X$, $x\preceq y$ and $p\in[0,1]$ we have 
		\[
			x=x\oplus_px\preceq x\oplus_py\preceq y\oplus_py=y.
		\]
		The second statement follows by parametric commutativity.
	\item Assume we have $x,y\in X$ with $x\preceq y$, and $p,p'\in[0,1]$ with $p\leq p'$. If $p=p'$ the assertion is 
		trivial, hence assume that $p<p'$. Then, in particular, $p<1$. 

		Set $r\DE\frac{p'-p}{1-p}$, then $r\in(0,1]$ and $r+(1-r)p=p'$. Hence, we obtain 
		\[
			y\oplus_{p'}x=y\oplus_r(y\oplus_px)\succeq y\oplus_px.
		\]
	\end{Elist}
\end{proof}

\begin{corollary}
\label{E4}
	Let $\langle X,\oplus_p,\preceq\rangle$ be a monotone convex algebra.
	Let $x,y\in X$ with $x\preceq y$. Then the map $\Gamma_{x,y}\DF[0,1]\to X$ from \eqref{E29} is either increasing 
	(in particular injective), or $\Gamma_{x,y}|_{(0,1)}$ is constant. 
\end{corollary}
\begin{proof}
	By Lemma~\ref{E2}(iii) the map $\Gamma_{x,y}$ is nondecreasing. Its kernel is one of the five congruences of 
	$\langle [0,1],+_p\rangle$ listed in Remark~\ref{E30}. 
\end{proof}

\subsection{Semicontinuous operations}

We study convex algebras whose underlying set is the interval $[0,1]$, and whose convex operations $\oplus_p$ enjoy certain 
(semi-) continuity properties. 
\begin{Itemize}
\item {\bf Convention 1:} 
	Throughout the paper all topological notions on $[0,1]$ refer to the Euclidean topology, that is, the
	topology inherited from the Euclidean metric $d(x,y)\DE|x-y|$ on $\bb R$. 
\item {\bf Convention 2:}  
	Throughout the paper we denote by ``$\leq$'' the usual order on $[0,1]$. All order theoretic terms, in
	particular the condition {\sf(MO)}, refer to this order. 
\end{Itemize}
We point out that the operations $\oplus_p$ will usually not be equal to the operations $+_p$ on $[0,1]$
induced by linear combinations. Of course, $\langle[0,1],+_p\rangle$ provides an example.
To recap: topology and order on $[0,1]$ are always the usual ones, while convex operations $\oplus_p$ are not
(except in a single particular case).

\begin{definition}
\label{E33}
	Let $\langle[0,1],\oplus_p\rangle$ be a convex algebra. We consider the following properties:
	\begin{itemize}
	\item[{\sf(UC)}]
		${\displaystyle
			\forall x,y\in[0,1)\DQ\forall p\in(0,1)\DP
			\limsup_{\varepsilon\to 0+}\big[(x+\varepsilon)\oplus_p(y+\varepsilon)\big]\leq x\oplus_py
		}$
	\item[{\sf(LC)}]
		${\displaystyle
			\forall x,y\in[0,1]\DP x\leq y\ \Rightarrow\ 
			\liminf_{p\to 1-}\big[y\oplus_px\big]\geq y
		}$
	\end{itemize}
\end{definition}

\noindent
The terminology ``{\sf(UC)}'' and ``{\sf(LC)}'' stems from the fact that these properties express a kind of semicontinuity 
of the map 
\[
	(x,y,p)\mapsto x\oplus_py.
\]
For {\sf(LC)} this is straightforward: for each fixed $x,y$ belonging to the triangle 
\[
	\big\{(x,y)\in[0,1]^2\DS x\leq y\big\}
\]
the function $p\mapsto y\oplus_px$ is lower semicontinuous at the
point $p=1$. For {\sf(UC)} the interpretation is a bit more loose: {\sf(UC)} matches
upper semicontinuity of $(x,y)\mapsto x\oplus_py$ for fixed $p$, but is one-sided since $\varepsilon$ is only allowed 
to approach $0$ from above, and restricted since the variables $x,y$ are not allowed to vary independently from
each other.

It does not come as a surprise that in conjunction with monotonicity {\sf(MO)} of operations {\sf(UC)} and {\sf(LC)} imply 
stronger continuity properties. 

\begin{lemma}
\label{E6}
	Let $\langle [0,1],\oplus_p\rangle$ be a monotone convex algebra.
	\begin{Enumerate}
	\item If $\langle[0,1],\oplus_p\rangle$ satisfies {\sf(UC)}, then 
		\begin{equation}
		\label{E7}
			\forall x,y\in[0,1)\DQ\forall p\in[0,1]\DP
			\lim_{\substack{x'\to x+\\ y'\to y+}}\big[x'\oplus_py'\big]=x\oplus_py,
		\end{equation}
		\begin{equation}
		\label{E8}
			\forall x,y\in[0,1],x\leq y\DQ\forall q\in(0,1)\DP
			\lim_{p\to q+}\big[y\oplus_px\big]=y\oplus_qx.
		\end{equation}
	\item If $\langle[0,1],\oplus_p\rangle$ satisfies {\sf(LC)}, then 
		\begin{equation}
		\label{E9}
			\forall x,y\in[0,1],x\leq y\DQ\forall q\in(0,1]\DP
			\lim_{p\to q-}\big[y\oplus_px\big]=y\oplus_qx.
		\end{equation}
	\end{Enumerate}
\end{lemma}
\begin{proof}
	Assume that {\sf(UC)} holds and $x,y\in[0,1)$. 
	If $p=0$ or $p=1$, the limit relation \eqref{E7} trivially holds. Assume that $p\in(0,1)$,
	and consider sequences $(x_n)_{n=1}^\infty$ with $x_n>x$ and $x_n\to x$, and $(y_n)_{n=1}^\infty$ with 
	$y_n>y$ and $y_n\to y$. Set 
	\[
		\varepsilon_n\DE\max\{x_n-x,y_n-y\},
	\]
	and note $\varepsilon_n>0$. 
	Then $\lim_{n\to\infty}\varepsilon_n=0$ and hence $\max\{x+\varepsilon_n,y+\varepsilon_n\}<1$ for all sufficiently
	large $n$. For such $n$ we have 
	\[
		x\oplus_py\leq x_n\oplus_py_n\leq(x+\varepsilon_n)\oplus_p(y+\varepsilon_n),
	\]
	and obtain 
	\[
		\limsup_{n\to\infty}(x_n\oplus_py_n)\leq\limsup_{n\to\infty}\Big[(x+\varepsilon_n)\oplus_p(y+\varepsilon_n)\Big]
		\stackrel{\text{\sf(UC)}}{\leq} x\oplus_py\leq\liminf_{n\to\infty}(x_n\oplus_py_n).
	\]
	Thus \eqref{E7} holds. 

	We come to the proof of \eqref{E8}. Let $x,y\in[0,1]$, $x\leq y$, and $q\in(0,1)$ be given. 
	The function $\Gamma_{x,y}$ is
	nondecreasing, and hence has at most countably many discontinuities. Choose $q'\in(q,1)$ such that $\Gamma_{x,y}$ is
	continuous at $q'$. Consider now a sequence $(p_n)_{n=1}^\infty$ with $p_n>q$ and $p_n\to q$. 
	Set $t_n\DE\frac{p_n}qq'$. Then $t_n>q'$ and $t_n\to q'$, in
	particular $t_n<1$ for all sufficiently large $n$, and for such $n$
	\[
		y\oplus_{t_n}x\geq y\oplus_{q'}x\quad\text{and}\quad\lim_{n\to\infty}(y\oplus_{t_n}x)=y\oplus_{q'}x.
	\]
	Hence, using  \eqref{E7},
	\[
		y\oplus_qx=(y\oplus_{q'}x)\oplus_{\frac q{q'}}x
		=\big[\lim_{n\to\infty}(y\oplus_{t_n}x)\big]\oplus_{\frac q{q'}}x
		=\lim_{n\to\infty}\big[(y\oplus_{t_n}x)\oplus_{\frac q{q'}}x\big]=\lim_{n\to\infty}(y\oplus_{p_n}x).
	\]
	Assume now that {\sf(LC)} holds; we have to prove \eqref{E9}. Since $x\leq y$ we have $y\oplus_qx\geq x$, and hence 
	\[
		y\oplus_qx\geq\limsup_{t\to 1-}\big[(y\oplus_qx)\oplus_tx\big]\geq 
		\liminf_{t\to 1-}\big[(y\oplus_qx)\oplus_tx\big]\stackrel{\text{\sf(LC)}}{\geq} y\oplus_qx.
	\]
	So, the limit exists.
	However, $(y\oplus_qx)\oplus_tx=y\oplus_{qt}x$. Setting $p\DE qt$, the limit $t\to 1-$ corresponds to the
	limit $p\to q-$. 
\end{proof}

\noindent
The relation \eqref{E7} says that for each fixed $p\in[0,1]$ the map $(x,y)\mapsto x\oplus_py$ is continuous from the right on the
domain $[0,1)^2$. The relations \eqref{E8} and \eqref{E9} together say that for each fixed $x,y\in[0,1]$, $x\leq y$, the map 
$\Gamma_{x,y}$ is continuous on the domain $(0,1]$. 

The bulk of the paper deals with convex algebras satisfying all three conditions {\sf(MO)}, {\sf(UC)}, {\sf(LC)}. 

\begin{definition}
\label{E10}
	We call $\langle [0,1],\oplus_p\rangle$ a \emph{monotone almost continuous convex algebra}, if it is a convex algebra that 
	satisfies {\sf(MO)}, {\sf(UC)}, {\sf(LC)}. 
\end{definition}

\noindent
It is interesting to observe how continuity properties influence the algebraic structure of the algebra. 
One instance is the following statement.

\begin{lemma}
\label{E11}
	Let $\langle [0,1],\oplus_p\rangle$ be a monotone convex algebra that satisfies {\sf(LC)}. Then a one-sided 
	cancellation law holds in $\langle[0,1],\oplus_p\rangle$. Namely, 
	\[
		\forall x,y,z\in[0,1]\DQ\forall p\in(0,1)\DP
		x\oplus_pz=y\oplus_pz\,\wedge\, z\leq x\,\wedge\, z\leq y\ \Rightarrow\ x=y.
	\]
\end{lemma}
\begin{proof}
	For every $p\in(0,1)$ we have 
	\[
		y\oplus_p0\leq y\oplus_pz=x\oplus_pz\leq x,\quad
		x\oplus_p0\leq x\oplus_pz=y\oplus_pz\leq y.
	\]
	Passing to the limit ``$p\to 1-$'' yields $y\leq x$ and $x\leq y$. 
\end{proof}

\section{Basic examples}
\label{E34}

Let us give some examples of monotone almost continuous convex algebra.
We point out that these are more than just examples: we will see in Section~\ref{E55} that they are the principle building blocks 
that every monotone almost continuous convex algebra is made of.

\begin{example}
\label{E23}
	Let $X$ be a convex subset of $\bb R$. 
	We already mentioned in Example~\ref{E32} that the convex operations $+_p$ on $X$ induced by linear combinations 
	satisfy {\sf(MO)}. Since addition and multiplication in $\bb R$ are continuous, the map 
	\[
		\left\{
		\begin{array}{rcl}
			X^2\times [0,1] & \to & [0,1]
			\\
			(x,y,p) & \mapsto & x+_py
		\end{array}
		\right.
	\]
	is continuous. 
	In particular, we see that $\langle[0,1],+_p\rangle$ is a monotone almost continuous convex algebra.
	For later reference we emphasize that 
	\begin{equation}
	\label{E51}
		\forall x,y\in X\DP x\leq y\Rightarrow\lim_{p\to 1-}[y+_px]=y,
	\end{equation}
	and if $X$ has more than one element, i.e., $X$ is an interval with nonempty interior, also
	(here $\sup X$ is understood in $\bb R\cup\{\infty\}$)
	\begin{equation}
	\label{E50}
		\forall x,y\in X\setminus\{\sup X\}\DQ\forall p\in(0,1)\DP
		\lim_{\varepsilon\to 0+}\big[(x+\varepsilon)+_p(y+\varepsilon)\big]=x+_py.
	\end{equation}
\end{example}

\noindent
The second example has a completely different algebraic behaviour.

\begin{example}
\label{E24}
	We define
	\[
		x\oplus_py\DE\max\{x,y\}\qquad\text{for }x,y\in X,p\in(0,1).
	\]
	It is straightforward to check that $\langle[0,1],\oplus_p\rangle$ is a monotone convex algebra. 
	The maps $\max\DF\bb R^2\to\bb R$, and with it also
	\[
		\left\{
		\begin{array}{rcl}
			[0,1]^2\times(0,1) & \to & [0,1]
			\\
			(x,y,p) & \mapsto & x\oplus_py
		\end{array}
		\right.
	\]
	are continuous and thus {\sf(UC)} certainly holds. For {\sf(LC)} we observe that 
	\[
		\left\{
		\begin{array}{rcl}
			\{(x,y)\in[0,1]^2\DS x\leq y\}\times(0,1] & \to & [0,1]
			\\
			(x,y,p) & \mapsto & y\oplus_px
		\end{array}
		\right.
	\]
	is continuous. 

	Note that the map $(x,y,p)\mapsto y\oplus_px$ is not continuous on all of $[0,1]^2\times[0,1]$. For example we
	have $\lim_{p\to 0+}1\oplus_p0=1\neq 0=1\oplus_00$.

	Also note that every subset $M\subseteq[0,1]$ is a subalgebra of $\langle[0,1],\oplus_p\rangle$. 
\end{example}

\noindent
In the third and fourth example we present algebras with mixed behaviour. 

\begin{example}
\label{E31}
	We define operations $+_p$, $p\in(0,1)$, on $[0,1]$ as 
	\[
		x\oplus_py\DE
		\begin{cases}
			px+(1-p)y \CAS x,y\in[0,1)
			\\
			1 \CAS x=1\vee y=1
		\end{cases}
	\]
	It is easy to check that $\langle[0,1],\oplus_p\rangle$ is a monotone almost continuous convex algebra.
	Again, the map $(x,y,p)\mapsto x\oplus_p y$ is not continuous, for example, we have 
	$\lim_{x\to 1-} x\oplus_{\frac 12}0=\frac 12\neq 1=1\oplus_p0$.
\end{example}

\begin{example}
\label{E35}
	Let operations $\oplus_p$ on $[0,1]$ be defined as
	\[
		x\oplus_py\DE 1-(1-x)^p(1-y)^{1-p}\qquad\text{for }x,y\in [0,1],p\in(0,1).
	\]
	Instead of working with $[0,1]$, $\oplus_p$, and $\leq$, it is in this example more practical (and more 
	enlightening) to pass to an isomorphic copy. Consider the set $[0,\infty]$, the operations $+_p$ for $p\in(0,1)$, 
	defined on $[0,\infty]$ by the usual computational conventions
	\[
		x+_py\DE
		\begin{cases}
			px+(1-p)y \CAS x,y\in[0,\infty)
			\\
			\infty \CASO
		\end{cases}
	\]
	and the usual order $\leq$ on $[0,\infty]$. Then $\langle[0,\infty],+_p,\leq\rangle$ is a monotone convex algebra.
	The functions
	\[
		\left\{
		\begin{array}{rcl}
			[0,\infty]^2\times(0,1) & \to & [0,1]
			\\
			(x,y,p) & \mapsto & x+_py
		\end{array}
		\right.
		\quad\text{and}\quad
		\left\{
		\begin{array}{rcl}
			\{(x,y)\in[0,\infty]^2\DS x\leq y\}\times(0,1] & \to & [0,1]
			\\
			(x,y,p) & \mapsto & y+_px
		\end{array}
		\right.
	\]
	are continuous. In particular, {\sf(UC)} and {\sf(LC)} hold. 

	Now let $f\DF[0,\infty]\to[0,1]$ be defined as 
	\begin{equation}
	\label{E44}
		f(t)\DE
		\begin{cases}
			1-e^{-t} \CAS t\in[0,\infty)
			\\
			1 \CAS t=\infty
		\end{cases}
	\end{equation}
	Then $f$ is an increasing continuous bijection, and 
	\[
		\forall x,y\in[0,1]\DQ\forall p\in(0,1)\DP 
		x\oplus_py=f\big(f^{-1}(x)+_pf^{-1}(y)\big).
	\]
	We see that $\langle[0,1],\oplus_p\rangle$ is a monotone almost continuous convex algebra.
\end{example}

\section{Classification of monotone almost continuous convex algebras}
\label{E55}

In this section we present our main results which establish a description of all monotone almost continuous convex algebras. This
description is achieved by putting together three theorems.
\begin{Itemize}
\item Theorem~\ref{E15}: Constructing monotone almost continuous convex algebras from certain data.
\item Theorem~\ref{E28}: Showing that every monotone almost continuous convex algebra can be obtained by the above construction.
\item Theorem~\ref{E40}: Describing which data gives rise to isomorphic algebras.
\end{Itemize}
The construction of algebras in Theorem~\ref{E15} is done by using isomorphic copies of the algebras from 
Examples~\ref{E23}, \ref{E24}, \ref{E31}, \ref{E35} as building blocks and plugging them together with a general construction method.
We start from data 
\[
	E,\ (\sigma_{a,b})_{(a,b)\in\Delta}
\]
where $E$ is a closed subset of $[0,1]$ with $0\in E$, 
\[
	\Delta\DE\big\{(a,b)\in[0,1]^2\DS a,b\in E,a<b,(a,b)\cap E=\emptyset\big\},
\]
and $\sigma_{a,b}\in\{1,\infty\}$ for all $(a,b)\in\Delta$. The plugging together process leading to a monotone almost continuous convex 
algebra $\langle[0,1],\oplus_p\rangle$ can be illustrated as
\begin{center}
\begin{tikzpicture}[x=1.2pt,y=1.2pt,scale=1.2,font=\fontsize{12}{12}]
	\draw[|-|] (0,0) -- (300,0);
	\draw (0,-15) node {$0$};
	\draw (300,-15) node {$1$};

	\draw[fill,color=purple] (0,0) circle [radius=1.2];
	\draw[fill,color=purple] (2,0) circle [radius=1.2];
	\draw[fill,color=purple] (5,0) circle [radius=1.2];
	\draw[fill,color=purple] (10,0) circle [radius=1.2];
	\draw[fill,color=purple] (30,0) circle [radius=1.2];
	\draw[fill,color=purple] (35,0) circle [radius=1.2];
	\draw[fill,color=purple] (38,0) circle [radius=1.2];
	\draw[thick,|-|,color=purple] (40,0) -- (70,0);
	\draw[color=purple] (70,-10) node {$\scriptstyle a$};
	\draw[thick,|-|,color=purple] (120,0) -- (170,0);
	\draw[color=purple] (120,-10) node {$\scriptstyle b$};
	\draw[fill,color=purple] (260,0) circle [radius=1.2];
	\draw[color=purple] (260,-10) node {$\scriptstyle\max E$};
	\draw[fill,color=purple] (260,0) circle [radius=1.2];
	\draw[fill,color=purple] (225,0) circle [radius=1.2];
	\draw[fill,color=purple] (220,0) circle [radius=1.2];
	\draw[fill,color=purple] (217,0) circle [radius=1.2];
	\draw[fill,color=purple] (215,0) circle [radius=1.2];
	\draw[color=purple] (215,-10) node {$\scriptstyle b'$};
	\draw[fill,color=purple] (195,0) circle [radius=1.2];
	\draw[color=purple] (195,-10) node {$\scriptstyle a'$};
	\draw[fill,color=purple] (180,0) circle [radius=1.2];
	\draw[fill,color=purple] (175,0) circle [radius=1.2];
	\draw[fill,color=purple] (172,0) circle [radius=1.2];

	\draw[color=purple] (0,15) node {\scalebox{0.8}{$E$ (purple)}};
	\draw[color=darkgray] (260,40) node {\scalebox{0.8}{$\langle[0,1],+_p\rangle$}};
	\draw[color=darkgray] (70,-50) node {\scalebox{0.8}{$\langle[0,\sigma_{a,b}),+_p\rangle$}};
	\draw[color=darkgray] (200,-50) node {\scalebox{0.8}{$\langle[0,\sigma_{a',b'}),+_p\rangle$}};
	\draw[color=darkgray,dotted,thick] (90,-25)--(110,-25);
	\draw[color=darkgray,dotted,thick] (50,-25)--(70,-25);
	\draw[color=darkgray,dotted,thick] (175,-25)--(195,-25);
	\draw[color=darkgray,dotted,thick] (210,-25)--(230,-25);
	\draw[thick,darkgray,->] (260,30) to[out=270,in=90] (280,3);
	\draw[thick,darkgray,->] (70,-40) to[out=90,in=270] (95,-3);
	\draw[thick,darkgray,->] (200,-40) to[out=90,in=270] (205,-3);
\end{tikzpicture}
\end{center}
Thereby elements $x$ of $E$ act below themselves as 
\[
	\forall y\in[0,x]\DQ\forall p\in(0,1)\DP x\oplus_py=x
\]
and elements $x$ in some interval $(a,b)$ with $(a,b)\in\Delta$ or in the possibly present interval $(a,1]$ with $a\DE\max E$ act
below that interval as 
\[
	\forall y\in[0,a]\DQ\forall p\in(0,1)\DP x\oplus_py=x\oplus_pa
\]

\subsection{Constructing monotone almost continuous convex algebras}

\begin{theorem}
\label{E15}
	Let $E$ be a closed subset of $[0,1]$ with $0\in E$, set
	\[
		\Delta\DE\{(a,b)\in[0,1]^2\DS a,b\in E,a<b,(a,b)\cap E=\emptyset\},
	\]
	and let $(\sigma_{a,b})_{(a,b)\in\Delta}\in\{1,\infty\}^\Delta$.
	Then there exists a monotone almost continuous convex algebra $\langle[0,1],\oplus_p\rangle$ such that 
	\begin{Enumerate}
	\item $E=\{y\in[0,1]\DS y\Eats 0\}$,
	\item For every $(a,b)\in\Delta$ the subalgebra $\langle[a,b),\oplus_p\rangle$ is isomorphic to 
		$\langle[0,\sigma_{a,b}),+_p\rangle$,
	\item If $\max E<1$, the subalgebra $\langle[\max E,1],\oplus_p\rangle$ is isomorphic to $\langle[0,1],+_p\rangle$.
	\end{Enumerate}
\end{theorem}

\noindent
Before we go into the details of the proof let us explain the main tool used in the construction,
namely Plonka sums. They go back to \cite{plonka:1967} and have been studied further, 
e.g.\ in \cite{plonka.romanowska:1992,fusco.paoli:2026}. 

Assume we are given 
\begin{Itemize}
\item a set $I\neq\emptyset$ that is endowed with a partial order such that each two elements have a least upper bound,
\item a family of algebras $A_i$, $i\in I$, that all have the same finitary signature,
\item a family of homomorphisms $\phi_{ij}\DF A_i\to A_j$, $i,j\in I$, $i\leq j$, such that 
	\[
		\forall i\in I\DP \phi_{ii}=\Id_{A_i}
		\quad\text{and}\quad
		\forall i,j,k\in I,i\leq j\leq k\DP \phi_{jk}\circ\phi_{ij}=\phi_{ik}.
	\]
\end{Itemize}
Let $A$ be the disjoint union 
\[
	A\DE\bigcup_{i\in I}\raisebox{11pt}{$\mkern-16mu{\scriptscriptstyle\bullet}\mkern16mu$} A_i,
\]
and define an algebra structure on $A$ as follows. Let $\tau$ be an $n$-ary operation symbol 
in the signature of the algebras $A_i$ and denote the corresponding operation on $A_i$ as $f_\tau^i$. 
Given $a_1,\ldots,a_{n}\in A$, let $i_1,\ldots,i_{n}\in I$ be the indices with $a_l\in A_{i_l}$, set 
$k\DE\sup\{i_1,\ldots,i_{n}\}$, and define 
\[
	f_\tau(a_1,\ldots,a_{n})\DE f_\tau^k\big(\phi_{i_1k}(a_1),\ldots,\phi_{i_{n}k}(a_{n})\big).
\]
The so obtained algebra $\langle A,(f_\tau)_\tau\rangle$ is called the \emph{Plonka sum} of the algebras 
$\langle A_i,(f_\tau^i)_\tau\rangle$, $i\in I$.

For convex algebras we can picture the construction of the Plonka sum as
\begin{center}
\begin{tikzpicture}[x=1.2pt,y=1.2pt,scale=0.8,font=\fontsize{12}{12}]
	\draw(200,90) node[anchor=south west] {$x\oplus_py\DE\phi_{ik}(x)\oplus_p\phi_{jk}(y)$};
	\draw(15,50) circle[radius=15];
	\draw(120,60) circle[radius=20];
	\draw(60,120) ellipse[x radius=50,y radius=30];
	\draw[->,dashed] (10,60) to[out=100,in=210] (30,120);
	\draw[->,dashed] (125,70) to[out=90,in=340] (90,110);

	\draw(0,90) node {\footnotesize$\phi_{ik}$};
	\draw(127,97) node {\footnotesize$\phi_{jk}$};
	\draw(110,140) node[anchor=west] {\small$A_k$\ \ (where $k=\sup\{i,j\}$)};
	\draw(0,45) node[anchor=east] {\small$A_i$};
	\draw(140,50) node[anchor=west] {\small$A_j$};

	\draw[fill] (11,57) circle [radius=1];
	\draw(16,54) node {\tiny$x$};
	\draw[fill] (33,122) circle [radius=1];
	\draw(33,129) node {\tiny$\phi_{ik}(x)$};
	\draw[fill] (125,67) circle [radius=1];
	\draw(131,65) node {\tiny$y$};
	\draw[fill] (87,112) circle [radius=1];
	\draw(87,119) node {\tiny$\phi_{jk}(x)$};

	\draw[dotted] (33,122) to (87,112);
	\draw[fill] (51,118.66) circle [radius=1.5];
	\draw(51,108) node {\footnotesize$x\oplus_py$};
\end{tikzpicture}
\end{center}
Plonka observed that every equation that holds in all algebras $A_i$ and has the property that the sets of variables 
occuring on the left and right sides coincide also holds in $A$. This applies in particular to convex algebras: the idempotence
law and the parametric commutative and associative laws are all of that kind. Thus every Plonka sum of convex algebras is again
a convex algebra.

\begin{proof}(of Theorem~\ref{E15})
	The proof is carried out in several steps. First we define an algebraic structure on $[0,1]$ by making a suitable Plonka
	sum. Then we check that it has all required properties.
\begin{Steps}
\item
	In this step we define the ingredients for a Plonka sum that will give an algebra structure on $[0,1]$. 
	The index set of this sum is $E$, and it is totally ordered with the usual order. 

	For each $a\in E$ we define a convex algebra $\langle X^{(a)},\oplus^{(a)}_p\rangle$. 
	The carrier set $X^{(a)}$ is 
	\begin{equation}
	\label{E53}
		X^{(a)}\DE
		\begin{cases}
			[a,b) \CAS (a,b)\in\Delta
			\\
			[a,1] \CAS a=\max E
			\\
			\{a\} \CASO
		\end{cases}
	\end{equation}
	Note that these case distinction could be formulated differently using that 
	\begin{align*}
		& (a,b)\in\Delta\ \Leftrightarrow\ E\cap(a,1]\neq\emptyset\wedge b\DE\inf\big(E\cap(a,1]\big)>a
		\\
		& a=\max E\ \Leftrightarrow\ E\cap(a,1]=\emptyset
	\end{align*}
	The set $X^{(a)}$ is endowed with operations $\oplus^{(a)}_p$ as follows. 
	In the first case in \eqref{E53} choose an increasing bijection $f_{a,b}\DF[0,\sigma_{a,b})\to[a,b)$ and transport the 
	operations induced by linear combinations on $[0,\sigma_{a,b})$ to $[a,b)$ with $f_{a,b}$:
	\[
		x\oplus_p^{(a)}y\DE f_{a,b}\Big(pf_{a,b}^{-1}(x)+(1-p)f_{a,b}^{-1}(y)\Big)
		\qquad\text{for }x,y\in X^{(a)},p\in(0,1).
	\]
	In the second case let $\oplus^{(a)}_p$ be the operations induced on $X^{(a)}$ by linear combinations. 
	In the last case $\oplus^{(a)}_p$ is trivial; the set $X^{(a)}$ has only one element. 

	For $a,a'\in E$ with $a<a'$ we define $\phi_{aa'}\DF X^{(a)}\to X^{(a')}$ as the constant map 
	\[
		\phi_{aa'}(x)\DE a'\qquad\text{for }x\in X^{(a)}.
	\]
	Clearly, $\phi_{aa'}$ is a homomorphism. Moreover, we set $\phi_{aa}\DE\Id_{X^{(a)}}$, $a\in E$. It is clear that 
	$\phi_{a'a''}\circ\phi_{aa'}=\phi_{aa''}$ for all $a\leq a'\leq a''$. 

	The sets $X^{(a)}$ are pairwise disjoint and their union is $[0,1]$. Hence, the Plonka sum of the algebras 
	$\langle X^{(a)},\oplus^{(a)}_p\rangle$, $p\in(0,1)$, gives operations $\oplus_p$ on $[0,1]$. 
\item
	We collect some properties of the algebras $\langle X^{(a)},\oplus^{(a)}\rangle$ and the homomorphisms $\phi_{aa'}$.
	\begin{Ilist}
	\item We have
		\begin{align*}
			& \forall a\in E\DP X^{(a)}\cap E=\{a\},
			\\
			& \forall x\in[0,1]\DP x\in X^{(\alpha(x))}\text{ where }\alpha(x)\DE\max\big(E\cap[0,x]\big).
		\end{align*}
		The function $\alpha\DF[0,1]\to E$ has several important properties. It is nondecreasing, it is continuous from
		the right (i.e.\ $\lim_{y\to x+}\alpha(y)=\alpha(x)$ for all $x\in[0,1)$), and 
		\[
			\forall x\in[0,1]\DP x\in E\ \Leftrightarrow\ x=\alpha(x)
		\]
	\item Not only the points of $E$ are totally ordered, but also the algebras $X^{(a)}$ are 
		ordered as whole blocks in the sense that 
		\[
			\forall a,a'\in E\DP a<a'\ \Leftrightarrow\ \Big(\forall x\in X^{(a)},y\in X^{(a')}\DP x<y\Big).
		\]
		Moreover, $a$ is the smallest element of $X^{(a)}$.
	\item The homomorphisms $\phi_{ij}$ have the monotonicity properties
		\begin{align*}
			& \forall x,y\in[0,1),x\leq y\DP \phi_{\alpha(x),\alpha(y)}(x)\leq y,
			\\
			& \forall a,c\in E,a\leq c\DQ\forall x,y\in X^{(a)},x\leq y\DP \phi_{ac}(x)\leq\phi_{ac}(y).
		\end{align*}
	\item By the definition of $\oplus^{(a)}_p$, every algebra $\langle X^{(a)},\oplus^{(a)}_p\rangle$ is isomorphic to a 
		convex subset of $\bb R$ with operations induced by linear combinations. The involved isomorphism 
		(that is: the maps $f_{a,b}$ and the identity map, respectively) are increasing bijections. Therefore, all
		those maps and their inverses are also continuous. 

		Now we recall Example~\ref{E23}, where we saw that convex subsets of $\bb R$ with the operations $+_p$ satisfy 
		{\sf(MO)} and the analogue \eqref{E51} of {\sf(LC)}. Provided that it has more than one element, it also
		satisfies the analogue \eqref{E50} of {\sf(UC)}. These properties are inherited
		by the isomorphisms, and hence every algebra $\langle X^{(a)},\oplus^{(a)}_p\rangle$ satisfies 
		{\sf(MO)} and \eqref{E51}. Provided that $X^{(a)}$ has more than one element, also \eqref{E50} holds.
	\item Being isomorphic to convex subsets of $\bb R$ with operations induced by linear combinations, 
		the algebras $\langle X^{(a)},\oplus^{(a)}_p\rangle$ are all cancellative. In particular, it follows that 
		\[
			\forall x,y\in X^{(a)},x\neq y\DP \Gamma_{x,y}\text{ injective}.
		\]
	\end{Ilist}
\item
	We check that $\langle[0,1],\oplus_p\rangle$ satisfies {\sf(MO)}.

	Let $x,x',y\in[0,1]$ with $x\leq x'$ and $p\in(0,1)$ be given. Denote 
	\[
		a\DE\alpha(x),\ a'\DE\alpha(x'),\ b\DE\alpha(y),\ c\DE\max\{a,b\},\ c'\DE\max\{a',b\}.
	\]
	Clearly, $a\leq a'$ and $c\leq c'$. By the definition of $\oplus_p$, it holds that 
	\[
		x\oplus_py=\phi_{ac}(x)\oplus_p^{(c)}\phi_{bc}(y),\quad
		x'\oplus_py=\phi_{a'c'}(x)\oplus_p^{(c')}\phi_{bc'}(y).
	\]
	We have $x\oplus_py\in X^{(c)}$ and $x'\oplus_py\in X^{(c')}$. If $c<c'$ it readily follows that 
	$x\oplus_py<x'\oplus_py$. Assume that $c=c'$. Then 
	\[
		\phi_{ac}(x)=\phi_{a'c}\big(\underbrace{\phi_{aa'}(x)}_{\leq x'}\big)\leq\phi_{a'c}(x').
	\]
	Since the algebra $\langle X^{(c)},\oplus_p^{(c)}\rangle$ satisfies {\sf(MO)}, this relation implies that 
	\[
		\phi_{ac}(x)\oplus_p^{(c)}\phi_{bc}(y)\leq\phi_{a'c}(x)\oplus_p^{(c)}\phi_{bc}(y).
	\]
\item
	We check that $\langle[0,1],\oplus_p\rangle$ satisfies {\sf(UC)}. 

	Let $x,y\in[0,1)$ and $p\in(0,1)$ be given. W.l.o.g.\ assume that $x\leq y$. We distinguish several cases. 
	\begin{Ilist}
	\item Assume that there exists $\varepsilon_n\downarrow 0$ such that 
		\[
			\forall n\in\bb N\DP \alpha(y+\varepsilon_{n+1})<\alpha(y+\varepsilon_n).
		\]
		We have 
		\[
			x\oplus_py=\phi_{\alpha(x)\alpha(y)}(x)\oplus_p^{(\alpha(y))}y\in X^{(\alpha(y))},
		\]
		and hence $x\oplus_py\geq\alpha(y)$. For $\varepsilon\in(0,\varepsilon_{n+1}]$ we have 
		\[
			\alpha(y+\varepsilon)\leq\alpha(y+\varepsilon_{n+1})<\alpha(y+\varepsilon_n).
		\]
		Since
		\[
			(x+\varepsilon)\oplus_p(y+\varepsilon)=
			\phi_{\alpha(x+\varepsilon)\alpha(y+\varepsilon)}(x)\oplus_p^{(\alpha(y+\varepsilon))}y
			\in X^{(\alpha(y+\varepsilon))},
		\]
		it follows that 
		\[
			(x+\varepsilon)\oplus_p(y+\varepsilon)<\alpha(y+\varepsilon_n).
		\]
		We conclude that 
		\[
			\forall n\in\bb N\DP \limsup_{\varepsilon\to 0+}\big[(x+\varepsilon)\oplus_p(y+\varepsilon)\big]
			\leq\alpha(y+\varepsilon_n).
		\]
		Since $\lim_{n\to\infty}\alpha(y+\varepsilon_n)=\alpha(y)$, therefore 
		\[
			\limsup_{\varepsilon\to 0+}\big[(x+\varepsilon)\oplus_p(y+\varepsilon)\big]\leq\alpha(y).
		\]
	\item Assume that there exists $\varepsilon_0>0$ such that 
		\[
			\alpha(x+\varepsilon_0)<\alpha(y),\quad
			\forall \varepsilon\in(0,\varepsilon_0]\DP\alpha(y+\varepsilon)=\alpha(y).
		\]
		By monotonicity then also 
		\[
			\forall\varepsilon\in[0,\varepsilon_0]\DP\alpha(x+\varepsilon)<\alpha(y).
		\]
		Moreover, we have $y,y+\varepsilon_0\in X^{(\alpha(y))}$, and hence $X^{(\alpha(y))}$ has more than one element.

		Since {\sf(MO)} and \eqref{E50} hold in the algebra $\langle X^{(\alpha(y))},\oplus_p^{(\alpha(y))}\rangle$, we
		obtain 
		\begin{align*}
			x\oplus_py= &\, 
			\phi_{\alpha(x)\alpha(y)}(x)\oplus_p^{(\alpha(y))}y=\alpha(y)\oplus_p^{(\alpha(y))}y
			=\lim_{\varepsilon\to 0+}\big[\alpha(y)\oplus_p^{(\alpha(y))}(y+\varepsilon)\big]
			\\
			= &\, \lim_{\varepsilon\to 0+}
			\big[\phi_{\alpha(x+\varepsilon)\alpha(y)}(x+\varepsilon)\oplus_p^{(\alpha(y))}(y+\varepsilon)\big]
			=\lim_{\varepsilon\to 0+}\big[(x+\varepsilon)\oplus_p(y+\varepsilon)\big].
		\end{align*}
	\item Assume that there exists $\varepsilon_0>0$ such that 
		\[
			\forall \varepsilon\in(0,\varepsilon_0]\DP\alpha(y+\varepsilon)=\alpha(y)=\alpha(x+\varepsilon).
		\]
		Then also $\alpha(x)=\alpha(y)$. Again $X^{(\alpha(y))}$ has more than one element, and hence \eqref{E50} holds in 
		$\langle X^{(\alpha(y))},\oplus_p^{(\alpha(y))}\rangle$. We obtain 
		\begin{align*}
			x\oplus_py= &\, x\oplus_p^{(\alpha(y))}y\geq
			\limsup_{\varepsilon\to 0+}\big[(x+\varepsilon)\oplus_p^{(\alpha(y))}(y+\varepsilon)\big]
			\\
			= &\, \limsup_{\varepsilon\to 0+}\big[(x+\varepsilon)\oplus_p(y+\varepsilon)\big].
		\end{align*}
	\end{Ilist}
\item
	We check that $\langle[0,1],\oplus_p\rangle$ satisfies {\sf(LC)}. 

	Let $x,y\in[0,1]$ with $x\leq y$ be given. Then $\alpha(x)\leq\alpha(y)$, and hence for
	each $p\in(0,1)$
	\[
		y\oplus_px=y\oplus_p^{(\alpha(y))}\phi_{\alpha(x)\alpha(y)}(x).
	\]
	Since $\phi_{\alpha(x)\alpha(y)}(x)\leq y$ and the algebra $\langle X^{(\alpha(y))},\oplus_p^{(\alpha(y))}\rangle$
	satisfies \eqref{E51}, it follows that 
	\[
		\liminf_{p\to 1-}\big[y\oplus_px\big]=
		\liminf_{p\to 1-}\big[y\oplus_p^{(\alpha(y))}\phi_{\alpha(x)\alpha(y)}(x)\big]\geq y.
	\]
\item
	We show (i)--(iii). 

	Items (ii) and (iii) hold by the definition of $\oplus_p$. For the proof of (i) denote the set on the right
	side of (i) as $E^\oplus$. 
	Let $y\in E$. If $y=0$, then trivially $y\in E^\oplus$. Assume that $y>0$, then $\alpha(0)=0<y=\alpha(y)$, and hence for
	each $p\in(0,1)$
	\[
		0\oplus_py=\phi_{\alpha(0)\alpha(y)}(0)\oplus_p^{(\alpha(y))}y=\alpha(y)\oplus^{(\alpha(y))}\alpha(y)
		=\alpha(y)=y.
	\]
	We see that $E\subseteq E^\oplus$. 

	Assume that $y\in[0,1]\setminus E$. Then $\alpha(y)<y$, and therefore $\Gamma_{\alpha(y),y}$ is injective. It follows
	that 
	\[
		\alpha(y)\oplus_{\frac 12}y=\alpha(y)\oplus_{\frac 12}^{(\alpha(y))}y\neq y,
	\]
	and we see that $y\notin E^\oplus$.
\end{Steps}
\end{proof}

\subsection{Structure of monotone almost continuous convex algebras}

We have already remarked that the relation $\Eats$ from Definition~\ref{E5} contains a lot of information about the algebraic
structure of a convex algebra. The conditions {\sf(MO)}, {\sf(UC)}, {\sf(LC)} have surprisingly strong
consequences on this relation. The following definitions are crucial for the further development.

\begin{definition}
\label{E17}
	Let $\langle[0,1],\oplus_p\rangle$ be a convex algebra. Then we denote 
	\[
		E^\oplus\DE\big\{y\in[0,1]\DS y\Eats 0\big\},
	\]
	and 
	\[
		V_{x,y}\DE\inf\big\{y\oplus_px\DS p\in(0,1]\big\}\qquad\text{for }x,y\in[0,1],x\leq y.
	\]
\end{definition}

\noindent
Note that 
\[
	\forall x,y\in[0,1],x\leq y\DP x\leq V_{x,y}\leq y\wedge\Big(V_{x,y}=y\Leftrightarrow y\Eats x\Big).
\]
Indeed, the inequalities follow directly from $x \le y\oplus_p x\le y$ and if $y \Eats x$ then immediately $V_{x,y} = y$. If $V_{x,y} = y$, then $y$ is a lower bound for $\{y\oplus_px\DS p\in(0,1]\big\}$ and hence $y \le y\oplus_p x \le y$. 

We next present two results which contain the technical core for the proof of Theorem~\ref{E28}.
The role of the number $V_{x,y}$ will become apparent from Lemma~\ref{E12}(iii) and Lemma~\ref{E13}(i).

\begin{lemma}
\label{E12}
	Let $\langle[0,1],\oplus_p\rangle$ be a monotone convex algebra. 
	\begin{Enumerate}
	\item If $\langle[0,1],\oplus_p\rangle$ satisfies {\sf(UC)}, then 
		${\displaystyle \forall x,y\in[0,1],x\leq y\DP V_{x,y}\Eats x}$.
	\item If $\langle[0,1],\oplus_p\rangle$ satisfies {\sf(LC)}, then 
		\begin{align*}
			& \forall x,y\in[0,1]\DP
			y\Eats x\ \Leftrightarrow\ \Big(x\leq y\wedge\Gamma_{x,y}\text{ not injective}\Big),
			\\
			& \forall y\in(0,1]\DP
			\Big(\exists x\in[0,1]\!\setminus\!\{y\}\DP y\Eats x\Big)\Leftrightarrow y\Eats 0
			\Leftrightarrow\Big(\forall x\in[0,1],x\leq y\DP y\Eats x\Big),
			\\[1mm]
			& \forall x,y\in[0,1],x\leq y\DP (V_{x,y},y)\cap E^\oplus=\emptyset.
		\end{align*}
	\item If $\langle[0,1],\oplus_p\rangle$ satisfies {\sf(UC)} and {\sf(LC)}, then $E^\oplus$ is closed and 
		\[
			\forall y\in[0,1]\DP \max(E^\oplus\cap[0,y])=V_{0,y}
		\]
	\end{Enumerate}
\end{lemma}
\begin{proof}
	\phantom{}
	\begin{Elist}
	\item We know that $x\leq V_{x,y}$ and, by monotonicity, that $V_{x,y}=\lim_{r\to 0+}(y\oplus_rx)$. 
		Hence \eqref{E7} yields
		\[
			V_{x,y}\geq V_{x,y}\oplus_{\frac 12}x=\big[\lim_{r\to 0+}(y\oplus_rx)\big]\oplus_{\frac 12}x
			=\lim_{r\to 0+}\big[(y\oplus_rx)\oplus_{\frac 12}x\big]
			=\lim_{r\to 0+}\big[y\oplus_{\frac r2}x\big]=V_{x,y}.
		\]
		Thus equality holds throughout, and we see that $V_{x,y}\Eats x$.
	\item \emph{First formula ``$\,\Rightarrow$'':}
		The function $\Gamma_{x,y}$ is constant on $(0,1]$, in particular, not injective. If
		$y\leq x$, then
		\[
			x=\lim_{p\to 1-}(x\oplus_py)=\lim_{p\to 1-}y=y.
		\]
		Hence, in any case, $x\leq y$.
		\\[1mm]
		\emph{First formula ``\,$\Leftarrow$'':}
		Since the kernel of $\Gamma_{x,y}$ is one of the five congruences exhibited in Remark~\ref{E30} and 
		$\Gamma_{x,y}$ is not injective, we know that 
		$\Gamma_{x,y}|_{(0,1)}$ is constant. Hence, 
		\[
			y=\lim_{p\to 1-}(y\oplus_px)=y\oplus_{\frac 12}x\leq y.
		\]
		We see that equality must hold throughout, and hence $y\Eats x$. 
		\\[1mm]
		\emph{Second formula ``\,$\Leftarrow$'':}
		Both backwards implications are trivial.
		\\[1mm]
		\emph{Second formula; first ``\,$\Rightarrow$'':}
		Assume $x\neq y$ and $y\Eats x$. Then $x<y$, and since $\lim_{p\to 1-}(y\oplus_p0)=y$ we can choose 
		$p\in(0,1)$ with $y\oplus_p0\geq x$. It follows that 
		\[
			y=y\oplus_px\leq y\oplus_p(y\oplus_p0)=y\oplus_{p+(1-p)p}0\leq y.
		\]
		We see that equality must hold throughout, and since $p+(1-p)p\in(0,1)$ it follows that $y\Eats 0$. 
		\\[1mm]
		\emph{Second formula; second ``\,$\Rightarrow$'':}
		This holds simply because of monotonicity: let $x\leq y$, then 
		\[
			y=y\oplus_{\frac 12}0\leq y\oplus_{\frac 12}x\leq y,
		\]
		and hence $y\Eats x$.
		\\[1mm]
		\emph{Third formula:}
		Assume towards a contradiction that $x,y\in[0,1],x\leq y$ and $(V_{x,y},y)\cap E^\oplus\neq\emptyset$. 
		Then, in particular, $V_{x,y}<y$ and hence $\neg(y\Eats x)$. Thus $\Gamma_{x,y}$ is injective. 
		Choose $z\in(V_{x,y},y)\cap E^\oplus$ and $p,q\in(0,1)$ with 
		\[
			y\oplus_px\leq z\leq y\oplus_qx.
		\]
		Since $y\oplus_{\frac p2}x<y\oplus_px\leq z$, we obtain for any $r\in(0,1)$ that 
		\[
			z=z\oplus_r(y\oplus_{\frac p2}x)\leq(y\oplus_qx)\oplus_r(y\oplus_{\frac p2}x)
			=y\oplus_{rq+(1-r)\frac p2}x.
		\]
		Since $\lim_{r\to 0+}(rq+(1-r)\frac p2)=\frac p2$, we find $r\in(0,1)$ with $rq+(1-r)\frac p2<p$. For such
		$r$ we have 
		\[
			y\oplus_{rq+(1-r)\frac p2}x<y\oplus_px\leq z,
		\]
		and this is a contradiction.
	\item We show that $E^\oplus$ is closed.
		Let $y\in[0,1]$ and assume we have a sequence $(y_n)_{n=1}^\infty$ in $E^\oplus$ with
		$y_n\to y$ and $y_n>y$. By \eqref{E7} we have 
		\[
			y\oplus_{\frac 12}0=\big[\lim_{n\to\infty}y_n\big]\oplus_{\frac 12}0
			=\lim_{n\to\infty}\big[y_n\oplus_{\frac 12}0\big]=\lim_{n\to\infty}y_n=y,
		\]
		and hence $y\Eats 0$. Assume we have a sequence $(y_n)_{n=1}^\infty$ in $E^\oplus$ with $y_n\to y$ and $y_n<y$. 
		By the third formula in (ii) we must have $y_n\leq V_{0,y}$, and conclude that $V_{0,y}=y$. This means that 
		$y\Eats 0$.

		To prove the second assertion, let $y\in[0,1]$. If $y\in E^\oplus$, then $V_{0,y}=y=\max(E^\oplus\cap[0,y])$. 
		Assume that $y\notin E^\oplus$. By (i) we have $V_{0,y}\in E^\oplus$ and by the third formula in (ii) we have 
		$(V_{0,y},y)\cap E^\oplus=\emptyset$. It again follows that $V_{0,y}=\max(E^\oplus\cap[0,y])$. 
	\end{Elist}
\end{proof}

\noindent
We also obtain important information about the complement of $E^\oplus$. 

\begin{lemma}
\label{E13}
	Let $\langle[0,1],\oplus_p\rangle$ be a monotone convex algebra.
	\begin{Enumerate}
	\item If $\langle[0,1],\oplus_p\rangle$ satisfies {\sf(UC)}, then
		\[
			\forall x,y\in[0,1]\DP\Big(x\leq y\Rightarrow\forall z\in[x,V_{x,y}]\DQ\forall p\in(0,1)\DP
			y\oplus_pz=y\oplus_pV_{x,y}\Big).
		\]
	\item If $\langle[0,1],\oplus_p\rangle$ satisfies {\sf(UC)} and {\sf(LC)}, then
		\[
			\forall x,y\in[0,1]\DP \Big(x<y\wedge(x,y]\cap E^\oplus=\emptyset\ \Rightarrow\ 
			\Gamma_{x,y}\DF[0,1]\to[x,y]\text{ is an increasing bijection}\Big).
		\]
	\end{Enumerate}
\end{lemma}
\begin{proof}
	\phantom{}
	\begin{Elist}
	\item Let $z\in[x,V_{x,y}]$ and $p\in(0,1)$. Then we have for every $r\in(0,1)$ 
		\[
			y\oplus_px\leq y\oplus_pz\leq y\oplus_pV_{x,y}
			\leq y\oplus_p(y\oplus_rx)=y\oplus_{p+(1-p)r}x.
		\]
		Since $\lim_{r\to 0+}(p+(1-p)r)=p$, we obtain from \eqref{E8} that 
		\[
			\lim_{r\to 0+}\big[y\oplus_{p+(1-p)r}x\big]=y\oplus_px.
		\]
	\item Since $V_{x,y}\in[x,y]$ and $V_{x,y}\Eats x$, we have $V_{x,y}=x$. Since $\neg(y\Eats x)$ we know that 
		$\Gamma_{x,y}$ is injective, hence increasing. 
		Since $\Gamma_{x,y}$ is continuous on $(0,1]$ and $x=V_{x,y}=\lim_{p\to 0+}(y\oplus_px)$, the map $\Gamma_{x,y}$ is
		continuous on all of $[0,1]$. It follows that $\Gamma_{x,y}\DF[0,1]\to[x,y]$ is also surjective.
	\end{Elist}
\end{proof}

\noindent
Similarly as in Lemma~\ref{E18} we can use Lemma~\ref{E13}(i) to obtain a relation for sums with more than two summands. 
The proof is again carried out using the parametric commutativity and associativity laws; we skip the details.

\begin{corollary}
\label{E39}
	Let $\langle[0,1],\oplus_p\rangle$ be a monotone convex algebra that satisfies {\sf(UC)}. 
	Let $A\neq\emptyset$, $(x_a)_{a\in A}\in[0,1]^A$, and $(p_a)_{a\in A}\in\mc DA$. Moreover, let $a_0\in A$ and assume
	that $p_{a_0}>0$. Then 
	\[
		\bigoplus_{a\in A}p_ax_a=\bigoplus_{a\in A}p_a\max\big\{x_a,V_{0,x_{a_0}}\}.
	\]
\end{corollary}

\noindent
Lemma~\ref{E13}(ii) enables us to construct a family of ''structural constants'' for a monotone almost continuous convex algebra.
More precisely, we can fully identify the algebraic behaviour on intervals containing no point of $E^\oplus$ in their interior.
Here we use the following notation.

\begin{definition}
\label{E42}
	Let $\langle[0,1],\oplus_p\rangle$ be a monotone almost continuous convex algebra. 
	Then we denote 
	\[
		\Delta^\oplus\DE\big\{(a,b)\in[0,1]^2\DS a,b\in E^\oplus,a<b,(a,b)\cap E^\oplus=\emptyset\big\}.
	\]
\end{definition}

\begin{theorem}
\label{E28}
	Let $\langle[0,1],\oplus_p\rangle$ be a monotone almost continuous convex algebra.
	\begin{Enumerate}
	\item Assume that $\max E^\oplus<1$. Then $\langle[\max E^\oplus,1],\oplus_p\rangle$ is isomorphic to 
		$\langle[0,1],+_p\rangle$. There exist exactly two isomorphisms, one of which is increasing and one decreasing.
	\item For each $(a,b)\in \Delta^\oplus$ there exists a unique element $\tau^\oplus_{a,b}\in\{1,\infty\}$ 
		such that $\langle[a,b),\oplus_p\rangle$ is isomorphic to $\langle[0,\tau^\oplus_{a,b}),+_p\rangle$. 
		Every isomorphism of $\langle[a,b),\oplus_p\rangle$ onto $\langle[0,\tau^\oplus_{a,b}),+_p\rangle$ is increasing.
		If $\tau^\oplus_{a,b}=1$, the isomorphism is unique. 
	\end{Enumerate}
\end{theorem}
\begin{proof}
	We show item (i); assume that $a\DE\max E^\oplus<1$. By Lemma~\ref{E13}(ii) the map $\Gamma_{a,1}$ is an increasing 
	isomorphism of $\langle[0,1],+_p\rangle$ onto $\langle[a,1],\oplus_p\rangle$. Let $\Phi$ be any isomorphism of these
	algebras. Then $\Gamma_{a,1}^{-1}\circ\Phi$ is an automorphism of $\langle[0,1],+_p\rangle$. By Remark~\ref{E14} thus 
	either $\Phi(t)=\Gamma_{a,1}(t)$, $t\in[0,1]$, or $\Phi(t)=\Gamma_{a,1}(1-t)$, $t\in[0,1]$. 

	The proof of (ii) is slightly more technical, since we cannot work with the map $\Gamma_{a,b}$ directly. 
	The idea is to exhaust $[a,b)$ with smaller intervals $[a,y]$ on which the map $\Gamma_{a,y}$ can be used. 

	Let $(a,b)\in\Delta^\oplus$ be given. Choose a sequence $(y_k)_{k=1}^\infty$ with 
	\[
		a<y_1<y_2<\ldots\quad\text{and}\quad \lim_{k\to\infty}y_k=b.
	\]
	Then each map $\Gamma_{a,y_k}$ is an increasing isomorphism of $\langle[0,1],+_p\rangle$ onto 
	$\langle[a,y_k],\oplus_p\rangle$. We define sequences $(r_k)_{k=0}^\infty$ and $(t_k)_{k=1}^\infty$ by 
	\begin{align*}
		& r_0\DE 1,\qquad r_k\DE\big(\Gamma_{a,y_{k+1}}\big)^{-1}(y_k)\quad\text{for }k\geq 1,
		\\
		& t_k\DE\Big(\prod_{j=0}^{k-1}r_j\Big)^{-1}\quad\text{for }k\geq 1.
	\end{align*}
	Then $r_k\in(0,1)$ for all $k\geq 1$, $t_1=1$, and $(t_k)_{k=1}^\infty$ is an increasing sequence. 
	Set $t_\infty\DE\lim_{k\to\infty}t_k\in(1,\infty]$. 

	Plugging the definitions yields the diagram (recall the map $M_t$ from Remark~\ref{E14})
	\[
		\begin{tikzcd}[column sep=large]
			\langle[0,t_k],+_p\rangle \arrow[d,swap,"\subseteq"] 
			& \langle[0,1],+_p\rangle \arrow[d,"M_{r_k}"] \arrow[l,swap,"M_{t_k}"] \arrow[r,"\Gamma_{a,y_k}"]
			& \langle[a,y_k],\oplus_p\rangle \arrow[d,"\subseteq"] 
			\\
			\langle[0,t_{k+1}],+_p\rangle 
			& \langle[0,1],+_p\rangle \arrow[l,swap,"M_{t_{k+1}}"] \arrow[r,"\Gamma_{a,y_{k+1}}"] 
			& \langle[a,y_{k+1}],\oplus_p\rangle
		\end{tikzcd}
	\]
	To see the right wing of the diagram note that all involved maps are homomorphisms, and hence it suffices to check the 
	images of $0$ and $1$.
	It follows that an increasing isomorphism $F\DF\langle[0,t_\infty),+_p\rangle\to\langle[a,b),\oplus_p\rangle$ is
	well-defined by the requirement that 
	\[
		\forall k\geq 1\DP F|_{[0,t_k]}=\Gamma_{a,y_k}\circ(M_{t_k})^{-1}.
	\]
	If $t_\infty=\infty$, we set $\tau^\oplus_{a,b}\DE\infty$ and use $\Phi\DE F$. Otherwise, we rescale: 
	set $\tau^\oplus_{a,b}\DE 1$ and $\Phi\DE F\circ M_{t_\infty}|_{[0,1)}$.

	Assume now that we have $\tilde\tau^\oplus_{a,b}\in\{1,\infty\}$ and an isomorphism 
	$\tilde\Phi\DF\langle[0,\tilde\tau^\oplus_{a,b}),+_p\rangle\to\langle[a,b),\oplus_p\rangle$. 
	Then $\Phi^{-1}\circ\tilde\Phi$ is an isomorphism of $\langle[0,\tilde\tau^\oplus_{a,b}),+_p\rangle$ onto 
	$\langle[0,\tau^\oplus_{a,b}),+_p\rangle$. 
	Remark~\ref{E14} implies that $\tilde\tau^\oplus_{a,b}=\tau^\oplus_{a,b}$ and that 
	$\Phi^{-1}\circ\tilde\Phi$ is increasing. Since $\Phi$ is increasing, also $\tilde\Phi$ must be.
\end{proof}

\subsection{Describing isomorphy of algebras}

In our third theorem we describe how the set $E^\oplus$ and the structural constants $\tau^\oplus_{a,b}$ determine the algebra 
up to isomorphism. 

\begin{theorem}
\label{E40}
	Let $\langle[0,1],\oplus_p\rangle$ and $\langle[0,1],\oplus_p'\rangle$ be monotone almost continuous convex algebras.
	Then $\langle[0,1],\oplus_p\rangle$ and $\langle[0,1],\oplus_p'\rangle$ are isomorphic, if and only if 
	\begin{Enumerate}
	\item there exists an increasing bijection $\varphi\DF E^\oplus\to E^{\oplus'}$ with 
		(note here that $(\varphi\times\varphi)(\Delta^\oplus)=\Delta^{\oplus'}$)
		\begin{equation}
		\label{E45}
			\forall (a,b)\in\Delta^\oplus\DP \tau^\oplus_{a,b}=\tau^{\oplus'}_{\varphi(a),\varphi(b)}
		\end{equation}
		and
	\item $\max E^\oplus<1\ \Leftrightarrow\ \max E^{\oplus'}<1$.
	\end{Enumerate}
\end{theorem}
\begin{proof}
In Steps\,\ding{192} and \ding{193} of the proof our goal is to show that the properties (i) and (ii) of the theorem 
enable construction of an isomorphism.
Hence, in these two steps, assume we have an increasing bijection $\varphi\DF E^\oplus\to E^{\oplus'}$ that satisfies
\eqref{E45}, and that either $\max E^\oplus=\max E^{\oplus'}=1$, or $\max E^\oplus<1$ and $\max E^{\oplus'}<1$.
\begin{Steps}
\item
	In this step we construct a candidate for $\Phi$.
	Let $(a,b)\in\Delta^\oplus$ and denote $a'\DE\varphi(a)$, $b'\DE\varphi(b)$. Then $(a',b')\in\Delta^{\oplus'}$. 
	Theorem~\ref{E28}(ii) provides us with increasing isomorphisms 
	$f_{a,b}\DF\langle[0,\tau^\oplus_{a,b}),+_p\rangle\to\langle[a,b),\oplus_p\rangle$ and 
	$f'_{a',b'}\DF\langle[0,\tau^{\oplus'}_{a',b'}),+_p\rangle\to\langle[a',b'),\oplus_p'\rangle$.
	Since $\tau^\oplus_{a,b}=\tau^{\oplus'}_{a',b'}$, mutually inverse increasing isomorphisms 
	$\Phi_{a,b}$ and $\Phi_{a',b'}'$ are well-defined by the diagram
	\[
		\begin{tikzcd}[column sep=large]
			\langle[a,b),\oplus_p\rangle \arrow[rr,dashed,bend left=5,"\Phi_{a,b}"]
			&& \langle[a',b'),\oplus_p'\rangle \arrow[ll,dashed,bend left=5,"\Phi_{a',b'}'"]
			\\
			& \langle[0,\tau^\oplus_{a,b}),+_p\rangle \arrow[ul,"f_{a,b}"] \arrow[ur,swap,"f'_{a',b'}"] &
		\end{tikzcd}
	\]
	If $\max E^\oplus<1$, then also $\max E^{\oplus'}<1$, and Theorem~\ref{E28}(i) provides us with increasing isomorphisms 
	$f_1\DF\langle[0,1],+_p\rangle\to\langle[\max E^\oplus,1],\oplus_p\rangle$ and 
	$f'_1\DF\langle[0,1],+_p\rangle\to\langle[\max E^{\oplus'},1],\oplus_p'\rangle$.
	Again mutually inverse increasing isomorphisms $\Phi_{\max E^\oplus,1}$ and $\Phi_{\max E^{\oplus'},1}'$ 
	are well-defined by the diagram
	\[
		\begin{tikzcd}[column sep=large]
			\langle[\max E^\oplus,1],\oplus_p\rangle \arrow[rr,dashed,bend left=5,"\Phi_{\max E^\oplus,1}"]
			&& \langle[\max E^{\oplus'},1],\oplus_p'\rangle \arrow[ll,dashed,bend left=5,"\Phi_{\max E^{\oplus'},1}'"]
			\\
			& \langle[0,1],+_p\rangle \arrow[ul,"f_1"] \arrow[ur,swap,"f'_1"] &
		\end{tikzcd}
	\]
	We plug the maps $\varphi$ and $\Phi_{a,b}$ together: define $\Phi\DF[0,1]\to[0,1]$ as 
	\[
		\Phi(x)\DE
		\begin{cases}
			\varphi(x) \CAS x\in E^\oplus
			\\
			\Phi_{a,b}(x) \CAS (a,b)\in\Delta^\oplus,x\in[a,b)
			\\
			\Phi_{\max E^\oplus,1}(x) \CAS \max E^\oplus<1\ \text{and}\ x\in[\max E^\oplus,1]
		\end{cases}
	\]
	The map $\Phi$ is well-defined since 
	\[
		[0,1]=E^\oplus\cup\bigcup_{(a,b)\in\Delta^\oplus}[a,b)\cup
		\begin{cases}
			\emptyset \CAS \max E^\oplus=1
			\\
			[\max E^\oplus,1] \CAS \max E^\oplus<1
		\end{cases}
	\]
	for each $(a,b)\in\Delta^\oplus$ it holds that $\varphi(a)=a'$, and $\varphi(\max E^\oplus)=\max E^{\oplus'}$.

	In the same way a map $\Phi'\DF[0,1]\to[0,1]$ is well-defined by 
	\[
		\Phi'(x)\DE
		\begin{cases}
			\varphi^{-1}(x) \CAS x\in E^{\oplus'}
			\\
			\Phi_{a',b'}'(x) \CAS (a',b')\in\Delta^{\oplus'},x\in[a',b')
			\\
			\Phi_{\max E^{\oplus'},1}'(x) \CAS \max E^{\oplus'}<1\ \text{and}\ x\in[\max E^{\oplus'},1]
		\end{cases}
	\]
	Clearly, $\Phi$ and $\Phi'$ are inverses of each other.

\item
	In this step we show that the map $\Phi$ constructed above is an increasing isomorphism of $\langle[0,1],\oplus\rangle$ 
	onto $\langle[0,1],\oplus'\rangle$.

	The fact that $\Phi$ is increasing is easy to see from the definitions. 
	Let $x,y\in[0,1]$ with $x<y$ be given. We distinguish cases.
	\begin{Itemize}
	\item $x\in E^\oplus,y\in E^{\oplus'}$: Then 
		\[
			\Phi(x)=\varphi(x)<\varphi(y)=\Phi(y).
		\]
	\item $x\notin E^\oplus,y\in E^{\oplus'}$: Set 
		\[
			a\DE\max\big(E^\oplus\cap[0,x)\big),\quad 
			b\DE\min(E^\oplus\cap(x,1]).
		\]
		Then $b\leq y$, and hence 
		\[
			\Phi(x)=\Phi_{a,b}(x)<\varphi(b)\leq\varphi(y)=\Phi(y).
		\]
	\item $x\in E^\oplus,y\notin E^{\oplus'}$: Set 
		\[
			a\DE\max\big(E^\oplus\cap[0,y)\big),\quad 
			b\DE
			\begin{cases}
				\min(E^\oplus\cap(y,1]) \CAS E^\oplus\cap(y,1]\neq\emptyset
				\\
				1 \CASO
			\end{cases}
		\]
		Then $x\leq a$, and hence 
		\[
			\Phi(x)=\varphi(x)\leq\varphi(a)=\Phi_{a,b}(a)<\Phi_{a,b}(y)=\Phi(y).
		\]
	\item $x\notin E^\oplus,y\notin E^{\oplus'}$ and $(x,y)\cap E^\oplus=\emptyset$: Set 
		\[
			a\DE\max\big(E^\oplus\cap[0,x)\big),\quad 
			b\DE
			\begin{cases}
				\min(E^\oplus\cap(y,1]) \CAS E^\oplus\cap(y,1]\neq\emptyset
				\\
				1 \CASO
			\end{cases}
		\]
		Then $[x,y]\subseteq(a,b)$, and hence 
		\[
			\Phi(x)=\Phi_{a,b}(x)<\Phi_{a,b}(y)=\Phi(y).
		\]
	\item $x\notin E^\oplus,y\notin E^{\oplus'}$ and $(x,y)\cap E^\oplus=\emptyset$: Set 
		\[
			a_x\DE\max\big(E^\oplus\cap[0,x)\big),\quad 
			b_x\DE\min(E^\oplus\cap(x,1]),
		\]
		\[
			a_y\DE\max\big(E^\oplus\cap[0,y)\big),\quad 
			b_y\DE
			\begin{cases}
				\min(E^\oplus\cap(y,1]) \CAS E^\oplus\cap(y,1]\neq\emptyset
				\\
				1 \CASO
			\end{cases}
		\]
		Then $b_x\leq a_y$, and hence 
		\[
			\Phi(x)=\Phi_{a_x,b_x}(x)<\varphi(b_x)\leq\varphi(a_y)=\Phi_{a_y,b_y}(a_y)<\Phi_{a_y,b_y}(y)=\Phi(y).
		\]
	\end{Itemize}
	To prove that $\Phi$ is a homomorphism it is enough to check that 
	\[
		\forall x,y\in[0,1],x<y\DQ\forall p\in(0,1)\DP \Phi(y\oplus_p x)=\Phi(y)\oplus_p'\Phi(x).
	\]
	Hence, let $x,y\in[0,1]$ with $x<y$, and $p\in(0,1)$ be given. Again we distinguish cases.
	\begin{Itemize}
	\item $y\in E^\oplus$: Then $\Phi(y)=\varphi(y)\in E^{\oplus'}$, and $\Phi(x)<\Phi(y)$, and $y\oplus_p x=y$. 
		Hence, 
		\[
			\Phi(y)\oplus_p'\Phi(x)=\Phi(y)=\Phi(y\oplus_p x).
		\]
	\item $y\notin E^\oplus$: Set 
		\[
			a\DE\max\big(E^\oplus\cap[0,y)\big),\quad 
			b\DE
			\begin{cases}
				\min(E^\oplus\cap(y,1]) \CAS E^\oplus\cap(y,1]\neq\emptyset
				\\
				1 \CASO
			\end{cases}
		\]
		If $(x,y)\cap E^\oplus=\emptyset$, then $[x,y]\subseteq[a,b)$, and we can compute
		\[
			\Phi(y\oplus_p x)=\Phi_{a,b}(y\oplus_p x)=\Phi_{a,b}(y)\oplus_p'\Phi_{a,b}(x)=\Phi(y)\oplus_p'\Phi(x).
		\]
		If $(x,y)\cap E^\oplus\neq\emptyset$, then $x\leq a=V_{0,y}$ and $\Phi(x)\leq\Phi(a)=V_{0,\Phi(y)}$, 
		and we can compute 
		\begin{align*}
			\Phi(y\oplus_p x)= &\, \Phi(y\oplus_p a)=\Phi_{a,b}(y\oplus_p a)
			\\
			= &\, \Phi_{a,b}(y)\oplus_p'\Phi_{a,b}(a)=\Phi(y)\oplus_p'\Phi(a)=\Phi(y)\oplus_p'\Phi(x).
		\end{align*}
	\end{Itemize}

	\medskip
	We come to the proof of the converse implication.
\item
	Assume that $\Phi\DF\langle[0,1],\oplus_p\rangle\to\langle[0,1],\oplus_p'\rangle$ is an isomorphism. 
	We start with showing that 
	\begin{equation}
	\label{E46}
		\Phi\big(E^\oplus\setminus\{0\}\big)=E^{\oplus'}\setminus\{0\},
	\end{equation}
	\begin{equation}
	\label{E47}
		\forall y\in E^\oplus\setminus\{0\}\DP \Phi\big([0,y)\big)=[0,\Phi(y)).
	\end{equation}
	Let $x,y\in[0,1]$. Since $\Phi$ is an isomorphism, we have 
	\[
		y\Eats x\ \Leftrightarrow\ y\oplus_{\frac 12}x=y\ \Leftrightarrow\ \Phi(y)\oplus_{\frac 12}\Phi(x)=\Phi(y)
		\ \Leftrightarrow\ \Phi(y)\Eats\Phi(x)
	\]
	Let $y\in E^\oplus\setminus\{0\}$ and $x\in[0,y)$. Then $y\Eats x$ and $y\neq x$, and hence also 
	$\Phi(y)\Eats\Phi(x)$ and $\Phi(y)\neq\Phi(x)$. Thus $\Phi(y)\in E^{\oplus'}\setminus\{0\}$ and $\Phi(x)<\Phi(y)$. 

	The map $\Phi$ is bijective, and we conclude from \eqref{E47} that also
	\begin{equation}
	\label{E48}
		\begin{aligned}
			\forall y\in E^\oplus\setminus\{0\}\DP &\, 
			\Phi\big([0,y]\big)=[0,\Phi(y)]\wedge
			\\
			&\, \Phi\big([y,1]\big)=[\Phi(y),1]\wedge \Phi\big((y,1]\big)=(\Phi(y),1],
		\end{aligned}
	\end{equation}
	\begin{equation}
	\label{E49}
		\forall x,y\in E^\oplus\setminus\{0\},x<y\DP\Phi\big([x,y)\big)=[\Phi(x),\Phi(y)).
	\end{equation}
	Using these relations, we already obtain some pieces of the required assertions. 
	The relation \eqref{E46} implies that $\Phi|_{E^\oplus\setminus\{0\}}$ is a bijection of $E^\oplus\setminus\{0\}$ onto
	$E^{\oplus'}\setminus\{0\}$. In particular, we have 
	\[
		E^\oplus=\{0\}\ \Leftrightarrow\ E^{\oplus'}=\{0\}
	\]
	Moreover, \eqref{E47} implies that the restriction $\Phi|_{E^\oplus}$ is increasing. 

	We can also deduce item (ii) of the theorem. Assume that $\max E^\oplus=1$. Then 
	\[
		[0,1]=\Phi\big([0,1]\big)=\Phi\big([0,\max E^\oplus]\big)=[0,\Phi(\max E^\oplus)]
		,
	\]
	and thus $\Phi(\max E^\oplus)=1$. Since $\Phi(\max E^\oplus)\in\Phi(E^\oplus\setminus\{0\})\subseteq E^{\oplus'}$, 
	we see that $\max E^{\oplus'}=1$. To show the converse implication, apply the already proven to $\Phi^{-1}$.
\item
	We settle the case that $E^\oplus=E^{\oplus'}=\{0\}$. 
	It is clear that in this case the statements Theorem~\ref{E40}(i),(ii) hold (note
	that $\Delta^\oplus=\Delta^{\oplus'}=\emptyset$). 
	Let $f\DF\langle[0,1],+_p\rangle\to\langle[0,1],\oplus_p\rangle$ and 
	$f'\DF\langle[0,1],+_p\rangle\to\langle[0,1],\oplus_p'\rangle$ be the increasing isomorphisms that exist uniquely 
	according to Theorem~\ref{E28}(i).
	Then a map $\Phi\DF[0,1]\to[0,1]$ is an isomorphism of $\langle[0,1],\oplus_p\rangle$ onto 
	$\langle[0,1],\oplus_p'\rangle$ if and only if $(f')^{-1}\circ\Phi\circ f$ is an automorphism of 
	$\langle[0,1],+_p\rangle$. These automorphisms are known from Remark~\ref{E14}; they are the maps 
	$t\mapsto t$ and $t\mapsto 1-t$.

	\medskip
	From now on we shall assume that $E^\oplus\neq\{0\}$ (and hence also $E^{\oplus'}\neq\{0\}$). 
\item
	Throughout Steps\,\ding{196}\,--\,\ding{199} assume we are given an isomorphism 
	$\Phi\DF\langle[0,1],\oplus_p\rangle\to\langle[0,1],\oplus_p'\rangle$.
	In this step we show that $\Phi(0)=0$. 

	The relation \eqref{E47} implies that $\Phi(0)\leq\Phi(y)$ for every $y\in E^\oplus$, and since 
	$\Phi(E^\oplus\setminus\{0\})=E^{\oplus'}\setminus\{0\}$ we obtain 
	\[
		\Phi(0)\leq\inf\big(E^{\oplus'}\setminus\{0\}\big).
	\]
	If the infimum equals $0$, we are done. Assume that $b'\DE\inf(E^{\oplus'}\setminus\{0\})>0$. Since $E^{\oplus'}$ is
	closed, we have $b'\in E^{\oplus'}\setminus\{0\}$. It follows that $b\DE\Phi^{-1}(b')$ is the smallest element of 
	$E^\oplus\setminus\{0\}$. Thus $(0,b)\in\Delta^\oplus$ and $(0,b')\in\Delta^{\oplus'}$, and Theorem~\ref{E28}(ii) provides us
	with increasing isomorphisms $f\DF\langle[0,\tau^\oplus_{0,b}),+_p\rangle\to\langle[0,b),\oplus_p\rangle$ and 
	$f'\DF\langle[0,\tau^{\oplus'}_{0,b'}),+_p\rangle\to\langle[0,b'),\oplus_p'\rangle$.
	Again referring to \eqref{E47} we find that $\Phi([0,b))=[0,b')$, and hence a map $\psi$ is well-defined 
	by the following diagram and is an isomorphism:
	\[
		\begin{tikzcd}[column sep=large]
			\langle[0,b),\oplus_p\rangle \arrow[r,"\Phi"]
			& \langle[0,b'),\oplus_p'\rangle 
			\\
			\langle[0,\tau^\oplus_{0,b}),+_p\rangle \arrow[u,"f"] \arrow[r,dashed,swap,"\psi"]
			& \langle[0,\tau^{\oplus'}_{0,b'}),+_p\rangle \arrow[u,swap,"f'"] 
		\end{tikzcd}
	\]
	Remark~\ref{E14} implies that $\tau^\oplus_{0,b}=\tau^{\oplus'}_{0,b'}$ and that $\psi$ is increasing. 
	In turn $\Phi|_{[0,b)}$ is increasing, which yields in particular that $\Phi(0)=0$.

	Now we know that $\Phi(E^\oplus)=E^{\oplus'}$, and $\Phi|_{E^\oplus}\DF E^\oplus\to E^{\oplus'}$ is an increasing 
	bijection. Hence, $\Phi|_{E^\oplus}$ qualifies as a candidate for a map $\varphi$ in (ii).
\item
	Let $(a,b)\in\Delta^\oplus$. We show that \eqref{E45} holds and $\Phi|_{[a,b)}$ is increasing. 

	If $a=0$, this was already seen in the previous step of the proof. 
	The analogous argument applies in the case that 
	$a>0$. Since $\Phi|_{E^\oplus}\DF E^\oplus\to E^{\oplus'}$ is an increasing bijection,
	we have $(\Phi(a),\Phi(b))\in\Delta^{\oplus'}$. Denote $a'\DE\Phi(a)$ and $b'\DE\Phi(b)$. Theorem~\ref{E28}(ii) gives 
	increasing isomorphisms $f\DF\langle[0,\tau^\oplus_{a,b}),+_p\rangle\to\langle[a,b),\oplus_p\rangle$ and 
	$f'\DF\langle[0,\tau^{\oplus'}_{a',b'}),+_p\rangle\to\langle[a,b'),\oplus_p'\rangle$.
	The relation \eqref{E49} tells us that $\Phi([a,b))=[a',b')$, and an isomorphism $\psi$ is well-defined 
	by the diagram:
	\[
		\begin{tikzcd}[column sep=large]
			\langle[a,b),\oplus_p\rangle \arrow[r,"\Phi"]
			& \langle[a',b'),\oplus_p'\rangle 
			\\
			\langle[0,\tau^\oplus_{a,b}),+_p\rangle \arrow[u,"f"] \arrow[r,dashed,swap,"\psi"]
			& \langle[0,\tau^{\oplus'}_{a',b'}),+_p\rangle \arrow[u,swap,"f'"] 
		\end{tikzcd}
	\]
	Remark~\ref{E14} implies that $\tau^\oplus_{a,b}=\tau^{\oplus'}_{a',b'}$ and that $\psi$ is increasing. In turn also 
	$\Phi|_{[a,b)}$ is increasing.
\item
	Assume that $a\DE\max E^\oplus<1$ (equivalently, $a'\DE\max E^{\oplus'}<1$). We show that 
	$\Phi|_{[a,1]}$ is increasing. Again the analogous argument applies. We have $\Phi(a)=a'$, and by \eqref{E48} therefore 
	$\Phi([a,1])=[a',1]$. Theorem~\ref{E28}(i) gives increasing isomorphisms
	$f\DF\langle[0,1],+_p\rangle\to\langle[a,1],\oplus_p\rangle$ and 
	$f'\DF\langle[0,1],+_p\rangle\to\langle[a',1],\oplus_p'\rangle$. Now an isomorphism $\psi$ is well-defined 
	by the diagram:
	\[
		\begin{tikzcd}[column sep=large]
			\langle[a,1],\oplus_p\rangle \arrow[r,"\Phi"]
			& \langle[a',1],\oplus_p'\rangle 
			\\
			\langle[0,1],+_p\rangle \arrow[u,"f"] \arrow[r,dashed,swap,"\psi"]
			& \langle[0,1],+_p\rangle \arrow[u,swap,"f'"] 
		\end{tikzcd}
	\]
	Since $\Phi(a)=a'$, we must have $\psi(0)=0$, and Remark~\ref{E14} implies that $\psi=\Id$. In particular, $\psi$ and with it
	also $\Phi|_{[a,1]}$ is increasing.
\item
	It is easy to deduce that $\Phi$ is increasing. Let $x,y\in[0,1]$ with $x<y$ be given. If there exists 
	$z\in[x,y]\cap E^\oplus$, then \eqref{E48} implies that $\Phi(x)\leq\Phi(z)\leq\Phi(y)$ and since $\Phi$ is injective
	we have $\Phi(x)<\Phi(y)$. If $[x,y]\cap E^\oplus=\emptyset$, then either there exists $(a,b)\in\Delta^\oplus$ such that
	$[x,y]\subseteq(a,b)$, or $[x,y]\subseteq(\max E^\oplus,1]$. In both cases the previous steps yield that 
	$\Phi(x)<\Phi(y)$.
\end{Steps}
\end{proof}

\noindent
Let us explicitly mention that in the course of the proof we established some more detailed statements for parts of the
theorem:
\begin{Itemize}
\item If {\rm(i)} and {\rm(ii)} hold, and $\varphi$ is as in {\rm(i)}, then there exists an increasing isomorphism $\Phi$ of 
	$\langle[0,1],\oplus_p\rangle$ onto $\langle[0,1],\oplus_p'\rangle$ with $\Phi|_{E^\oplus}=\varphi$. 
\item If $E^\oplus=E^{\oplus'}=\{0\}$, there exist exactly two isomorphisms of $\langle[0,1],\oplus_p\rangle$ onto 
	$\langle[0,1],\oplus_p'\rangle$. One of them is increasing and the other is decreasing.
\item Assume that $E^\oplus\neq\{0\}$ and $\Phi\DF\langle[0,1],\oplus_p\rangle\to\langle[0,1],\oplus_p'\rangle$ is an 
	isomorphism. Then $\Phi$ is increasing, maps $E^\oplus$ bijectively onto $E^{\oplus'}$, and 
	$\varphi\DE\Phi|_{E^\oplus}$ satisfies \eqref{E45}. 
\end{Itemize}
Next, we revisit the examples given in Section~\ref{E34}.

\begin{remark}
\label{E41}
	We give for each of Examples~\ref{E23}, \ref{E24}, \ref{E31}, \ref{E35} 
	the respective set $E^\oplus$, structural constants 
	$(\tau^\oplus_{a,b})_{(a,b)\in\Delta^\oplus}$, and the automorphism group of $\langle[0,1],\oplus_p\rangle$ 
	(denote it as $\Aut\langle[0,1],\oplus_p\rangle$). 
	\begin{Itemize}
	\item In Example~\ref{E23}: 
		$E^\oplus=\{0\}$, $\Delta^\oplus=\emptyset$, and $\Aut\langle[0,1],\oplus_p\rangle=\{\Id,1-\Id\}$.
	\item In Example~\ref{E24}: $E^\oplus=[0,1]$, $\Delta^\oplus=\emptyset$, and 
		\\
		$\Aut\langle[0,1],\oplus_p\rangle=\{\Phi\DF[0,1]\to[0,1]\DS \Phi\text{ increasing bijection}\}$.
	\item In Example~\ref{E31}: $E^\oplus=\{0,1\}$, $\Delta^\oplus=\{(0,1)\}$, $\tau^\oplus_{0,1}=1$, and 
		$\Aut\langle[0,1],\oplus_p\rangle=\{\Id\}$.
	\item In Example~\ref{E35}: $E^\oplus=\{0,1\}$, $\Delta^\oplus=\{(0,1)\}$, $\tau^\oplus_{0,1}=\infty$, and 
		\\
		$\Aut\langle[0,1],\oplus_p\rangle=\big\{f\circ M_t\circ f^{-1}\DS t>0\big\}$ 
		where $f$ is the map from \eqref{E44} and $M_t$ is extended
		to $[0,\infty]$ by the usual convention that $t\cdot\infty=\infty$.
	\end{Itemize}
\end{remark}

\noindent
For the purpose of illustration we discuss two examples which are more complex than our building blocks 
Examples~\ref{E23}, \ref{E24}, \ref{E31}, \ref{E35}, and show mixed behaviour.

\begin{example}
\label{E56}
	Let $r\in(0,1)$, and consider the set 
	\[
		E_r\DE\big\{r^n\DSb n\in\bb N\big\}\cup\{0\}.
	\]
	Then $E_r$ is a closed subset of $[0,1]$ with $0\in E_r$. The set $\Delta_r$ corresponding to $E_r$ is 
	\[
		\Delta_r=\big\{(r^{n+1},r^n)\DSb n\in\bb N\big\},
	\]
	and we can identify it with $\bb N$. 
	For each sequence $(\sigma_n)_{n\in\bb N}\in\{1,\infty\}^{\bb N}$ Theorem~\ref{E15} provides us with a monotone almost 
	continuous algebra $\langle[0,1],\oplus_p\rangle$.
	\begin{center}
	\begin{tikzpicture}[x=1.2pt,y=1.2pt,scale=1.2,font=\fontsize{10}{10}]
		\draw[|-|] (0,0) -- (200,0);
		\draw (200,-8) node {$1=r^0$};
		\draw (100,-8) node {$r^1$};
		\draw (50,-8) node {$r^2$};
		\draw (25,-8) node {$r^3$};
		\draw[loosely dotted,thick] (5,-9) to (20,-9);
		\draw (0,-9) node {$0$};

		\draw[fill,color=purple] (200,0) circle [radius=1.2];
		\draw[fill,color=purple] (100,0) circle [radius=1.2];
		\draw[fill,color=purple] (50,0) circle [radius=1.2];
		\draw[fill,color=purple] (25,0) circle [radius=1.2];
		\draw[fill,color=purple] (12.5,0) circle [radius=1.2];
		\draw[fill,color=purple] (0,0) circle [radius=1.2];
		\draw[color=purple,dotted,thick] (2,1) to (11,1);

		\draw(150,3) node[anchor=south] {$\sigma_0$};
		\draw(75,3) node[anchor=south] {$\sigma_1$};
		\draw(37.5,3) node[anchor=south] {$\sigma_2$};
		\draw[loosely dotted,thick] (5,8) to (30,8);
	\end{tikzpicture}
	\end{center}
	The isomorphy classes of those algebras can be determined based on Theorem~\ref{E40}. To this end observe that for each 
	$r,r'\in(0,1)$ there exists exactly one increasing bijection of $E_r$ onto $E_{r'}$. 
	Hence, two algebras constructed from $r,(\sigma_n)_{n\in\bb N}$ and $r',(\sigma_n')_{n\in\bb N}$ are isomorphic if and only 
	if $\sigma_n=\sigma_n'$ for all $n\in\bb N$.
	Two different algebras indeed may have a very different structure. To illustrate this observe that the automorphism group of
	the algebra obtained from $E_r,(\sigma_n)_{n\in\bb N}$ is isomorphic to $(\bb R^+)^{\{n\in\bb N\DS \sigma_n=\infty\}}$.
\end{example}

\noindent
In the above example the set $E$ enforces a certain stability, due to the fact that there are no nontrivial increasing bijections of $E$
onto itself. An example where $E$ admits nontrivial increasing bijections is the following.

\begin{example}
\label{E43}
	Fix a parameter $r\in(0,1)$, and consider the set 
	\[
		E\DE\big\{r^{2^n}\DSb n\in\bb Z\big\}\cup\{0,1\}.
	\]
	Then $E$ is a closed subset of $[0,1]$ with $0\in E$. The set $\Delta$ corresponding to $E$ is 
	\[
		\Delta=\big\{\big(r^{2^n},r^{2^{n-1}}\big)\DSb n\in\bb Z\big\},
	\]
	and we can identify it with $\bb Z$. 
	For each sequence $(\sigma_n)_{n\in\bb Z}\in\{1,\infty\}^{\bb Z}$ Theorem~\ref{E15} provides us with a monotone almost 
	continuous algebra $\langle[0,1],\oplus_p\rangle$.
	\begin{center}
	\begin{tikzpicture}[x=1.2pt,y=1.2pt,scale=1.2,font=\fontsize{10}{10}]
		\draw[|-|] (0,0) -- (200,0);
		\draw (200,-8) node {$1$};
		\draw[loosely dotted,thick] (195,-9) to (182,-9);
		\draw (175,-8) node {$r^{2^{-2}}$};
		\draw (150,-8) node {$r^{2^{-1}}$};
		\draw (100,-8) node {$r^{2^0}$};
		\draw (50,-8) node {$r^{2^1}$};
		\draw (25,-8) node {$r^{2^2}$};
		\draw[loosely dotted,thick] (5,-9) to (18,-9);
		\draw (0,-9) node {$0$};

		\draw[fill,color=purple] (200,0) circle [radius=1.2];
		\draw[color=purple,dotted,thick] (198,1) to (189,1);
		\draw[fill,color=purple] (187.5,0) circle [radius=1.2];
		\draw[fill,color=purple] (175,0) circle [radius=1.2];
		\draw[fill,color=purple] (150,0) circle [radius=1.2];
		\draw[fill,color=purple] (100,0) circle [radius=1.2];
		\draw[fill,color=purple] (50,0) circle [radius=1.2];
		\draw[fill,color=purple] (25,0) circle [radius=1.2];
		\draw[fill,color=purple] (12.5,0) circle [radius=1.2];
		\draw[color=purple,dotted,thick] (2,1) to (11,1);
		\draw[fill,color=purple] (0,0) circle [radius=1.2];

		\draw[loosely dotted,thick] (195,8) to (170,8);
		\draw(162.5,2) node[anchor=south] {$\sigma_{-1}$};
		\draw(125,3) node[anchor=south] {$\sigma_0$};
		\draw(75,3) node[anchor=south] {$\sigma_1$};
		\draw(37.5,3) node[anchor=south] {$\sigma_2$};
		\draw[loosely dotted,thick] (5,8) to (30,8);
	\end{tikzpicture}
	\end{center}
	The isomorphy classes of those algebras can be determined based on Theorem~\ref{E40}. To this end observe that the set of all
	increasing bijections of $E$ onto itself is $\{\varphi_m\DS m\in\bb Z\}$ where 
	\[
		\varphi_m(x)\DE x^{2^m}.
	\]
	Hence, two sequences $(\sigma_n)_{n\in\bb Z}$ and $(\tilde\sigma_n)_{n\in\bb Z}$ give rise to isomorphic algebras, if
	and only if 
	\[
		\exists m\in\bb Z\DQ\forall n\in\bb Z\DP \tilde\sigma_n=\sigma_{n+m}
	\]
	Again different algebras may have very different structure. For example:
	\begin{Itemize}
	\item If $\sigma_n=1$ for all $n\in\bb Z$, then $\Aut\langle[0,1],\oplus_p\rangle\cong\bb Z$.
	\item If $\sigma_0=\infty$ and $\sigma_n=1$ for $n\in\bb Z\setminus\{0\}$, then 
		$\Aut\langle[0,1],\oplus_p\rangle\cong\bb R^+$.
	\end{Itemize}
	If we change the parameter $r$, there exist increasing bijections between the corresponding sets $E$. Hence, isomorphy within
	all these algebras is decided in a similar way. 
\end{example}

\section{Semicontinuity of homomorphic extensions}
\label{E22}

In this section we make the connection with a semicontinuity property that appeared in the work of Mio 
(see the paragraph after the proof of Theorem~5.3 in \cite{mio:2025}).

\begin{definition}
\label{E38}
	Let $\langle[0,1],\oplus_p\rangle$ be a convex algebra. We consider the following property:
	\begin{itemize}
	\item[\LC] For every nonempty set $A$ and every map $\varphi\DF A\to[0,1]$ the extension $\varphi^\#$ of 
		$\varphi$ to a homomorphism of $\mc DA$ into $\langle[0,1],\oplus_p\rangle$ is lower semicontinuous. 
	\end{itemize}
	Here $\mc DA$ is endowed with the topology of pointwise convergence (i.e., the restriction of the product
	topology of $\bb R^A$ to $\mc DA$).
\end{definition}

\noindent
The condition \LC\ is clearly much stronger than {\sf(LC)}: if $x,y\in X$
with $x\leq y$ are given, then {\sf(LC)} states nothing but lower semicontinuity at the point $t\DE 1$ of the extension 
$\varphi^\#\DF\mc D(\{0,1\})\to X$ of the map $\varphi\DF\{0,1\}\to X$ defined as $\varphi(0)\DE x$, $\varphi(1)\DE y$.

Interestingly, in conjunction with {\sf(MO)} and {\sf(UC)} the conditions {\sf(LC)} and \LC\ are equivalent.

\begin{proposition}
\label{E26}
	Every monotone almost continuous convex algebra satisfies \LC.
\end{proposition}

\noindent
We isolate the source of semicontinuity that enables the argument.

\begin{lemma}
\label{E19}
	Let $\langle[0,1],\oplus_p\rangle$ be a monotone almost continuous convex algebra, and let $A$ be a
	nonempty set. Let $w\in X\setminus E^\oplus$, and $(\xi_a)_{a\in A}\in[V_{0,w},w]^A$. Then the map 
	\[
		\Lambda^\xi\DF\left\{
		\begin{array}{rcl}
			\mc DA & \to & [V_{0,w},w]
			\\
			(p_a)_{a\in A} & \mapsto & \bigoplus_{a\in A}p_a\xi_a
		\end{array}
		\right.
	\]
	is lower semicontinuous. 
\end{lemma}
\begin{proof}
	We use the increasing isomorphism $\Gamma\DE\Gamma_{V_{0,w},w}$ to transport the problem from 
	$\langle[V_{0,w},w],\oplus_p\rangle$ to $\langle[0,1],+_p\rangle$: set $t_a\DE\Gamma^{-1}(\xi_a)$, then we have 
	\[
		\forall(p_a)_{a\in A}\in\mc DA\DP
		\Lambda^\xi\big((p_a)_{a\in A}\big)=\Gamma\Big(\sum_{a\in A}p_at_a\Big).
	\]
	Since 
	\[
		\sum_{a\in A}p_at_a=\sup_{\substack{A'\subseteq A\\ A'\text{ finite}}}\sum_{a\in A'}p_at_a, 
	\]
	and each sum taken over a fixed finite subset $A'$ is a continuous function of $(p_a)_{a\in A}$, the function 
	mapping $(p_a)_{a\in A}$ to $\sum_{a\in A}p_at_a$ is lower semicontinuous. 

	Let $(q_a)_{a\in A}\in\mc DA$, $\gamma\in\bb R$, and assume that $\gamma<\Lambda^\xi((q_a)_{a\in A})$. We have to find a
	neighbourhood $U$ of $(q_a)_{a\in A}$ in $\mc DA$ such that this inequality holds throughout $U$. 
	If $\gamma<V_{0,w}$, we can choose $U\DE\mc DA$. Assume that $\gamma\geq V_{0,w}$. Since $\gamma$ is bounded from
	above by a value of $\Lambda^\xi$, it certainly does not exceed $w$. Hence, the value $\Gamma^{-1}(\gamma)$ is defined.
	Monotonicity of $\Gamma$ yields 
	\[
		\Gamma^{-1}(\gamma)<\sum_{a\in A}q_at_a,
	\]
	and we find a neighbourhood $U$ of $(q_a)_{a\in A}$ such that this inequality prevails throughout $U$. Applying
	$\Gamma$ yields the required assertion. 
\end{proof}

\begin{proof}(of Proposition~\ref{E26})
	Assume we have a nonempty set $A$, a map $\varphi\DF A\to[0,1]$, a point $(q_a)_{a\in A}$, and a number 
	$\gamma\in\bb R$ with $\gamma<\varphi^\#((q_a)_{a\in A})$. 

	Write $\{a\in A\DS q_a>0\}=\{a_1,\ldots,a_m\}$, where the enumeration is chosen such that 
	\[
		\varphi(a_1)=\max\big\{\varphi(a_j)\DS j\in\{1,\ldots,m\}\big\},
	\]
	and set 
	\[
		O\DE\big\{(p_a)_{a\in A}\in\mc DA\DS p_{a_1}>0\big\}.
	\]
	Then $O$ is an open subset of $\mc DA$ and $(q_a)_{a\in A}\in O$.

	By Corollary~\ref{E39} we have 
	\[
		\forall(p_a)_{a\in A}\in O\DP
		\varphi^\#\big((p_a)_{a\in A}\big)=\bigoplus_{a\in A}p_a\max\{\varphi(a),V_{0,\varphi(a_1)}\}.
	\]
	The case that $\varphi(a_1)\in E^\oplus$ is easily settled. In this case we have $V_{0,\varphi(a_1)}=\varphi(a_1)$, and
	hence can estimate for each $(p_a)_{a\in A}\in O$ 
	\begin{align*}
		\gamma< &\, 
		\varphi^\#\big((q_a)_{a\in A}\big)
		=\bigoplus_{j=1}^mq_{a_j}\underbrace{\max\{\varphi(a_j),\varphi(a_1)\}}_{=\varphi(a_1)}=\varphi(a_1)
		=\bigoplus_{a\in A}p_a\varphi(a_1)
		\\
		\leq &\, \bigoplus_{a\in A}p_a\max\{\varphi(a),\varphi(a_1)\}
		=\varphi^\#\big((p_a)_{a\in A}\big).
	\end{align*}
	Consider now the case that $\varphi(a_1)\notin E^\oplus$. The idea is to push $\varphi(a)$ into the interval 
	$[V_{0,\varphi(a_1)},\varphi(a_1)]$: set 
	\[
		\xi_a\DE\min\big\{\max\{\varphi(a),V_{0,\varphi(a_1)}\},\varphi(a_1)\big\}\qquad\text{for }a\in A.
	\]
	The map $\Lambda^\xi$ from Lemma~\ref{E19} satisfies
	\begin{align*}
		\Lambda^\xi\big((q_a)_{a\in A}\big)= &\, \bigoplus_{j=1}^mq_{a_j}\xi_{a_j}
		=\bigoplus_{j=1}^mq_{a_j}\min\big\{\max\{\varphi(a_j),V_{0,\varphi(a_1)}\},\varphi(a_1)\big\}
		\\
		= &\, \bigoplus_{j=1}^mq_{a_j}\max\{\varphi(a_j),V_{0,\varphi(a_1)}\}
		=\varphi^\#\big((q_a)_{a\in A}\big)>\gamma.
	\end{align*}
	Lemma~\ref{E19} tells us that there exists a neighbourhood $U$ of $(q_a)_{a\in A}$ such that 
	\[
		\forall (p_a)_{a\in A}\in U\DP \Lambda^\xi\big((p_a)_{a\in A}\big)>\gamma.
	\]
	The set $U\cap O$ is again a neighbourhood of $(q_a)_{a\in A}$, and for $(p_a)_{a\in A}\in U\cap O$ we can estimate 
	\begin{align*}
		\varphi^\#\big((p_a)_{a\in A}\big)= &\, 
		\bigoplus_{a\in A}p_a\max\{\varphi(a),V_{0,\varphi(a_1)}\}
		\\
		\geq &\, \bigoplus_{a\in A}p_a\min\big\{\max\{\varphi(a),V_{0,\varphi(a_1)}\},\varphi(a_1)\}
		=\Lambda^\xi\big((p_a)_{a\in A}\big)>\gamma.
		\\
	\end{align*}
\end{proof}

\section{Conclusions}

We gave an explicit construction of all convex algebras on $[0,1]$ with monotone and semicontinuous operations with respect to
the standard order and topology on the unit interval. The theory of generalized ICA for these convex algebras is compact. It would be
interesting to investigate if this or similar classes play a role in the proof of compactness of other quantitative theories.  

The set $[0,1]$ is the carrier of the free convex algebra $\mc D2$ with two generators. Monotone almost continuous convex algebras
with carrier set $\mc Dn$ where $n>2$ (or even more general carriers) could be defined in the very same way; simply use the product order 
and product topology. It may be subject of future work to understand the structure of such algebras. However, we expect the situation
to be much more complicated. A severe obstacle may arise from the fact that the congruence lattice of $\mc Dn$ is much more complex than
the one of $\mc D2$, and it is exactly this simplicity of $\mc D2$ that enables the present description via the set $E^\oplus$ of eaters.


\bibliography{ca123}

\end{document}